\documentclass[12pt]{article}
\RequirePackage[margin=1in]{geometry}
\usepackage[utf8]{inputenc}
\usepackage[sort&compress]{natbib}
\bibpunct{(}{)}{;}{a}{}{,} 
\usepackage[usenames,dvipsnames]{xcolor}
\usepackage{graphicx, subfigure}
\usepackage{float}
\usepackage{algorithm}
\usepackage{algpseudocode}
\usepackage{etoolbox}\AtBeginEnvironment{algorithmic}{\small}

\usepackage{amsmath} 
\usepackage{bm} 
\usepackage{amssymb} 
\usepackage{bbm} 
\usepackage{mathrsfs}
\usepackage{amsthm}
\usepackage{multirow}
\usepackage{url}

\newtheorem{theorem}{Theorem}
\newtheorem{cor}{Corollary}

\def\YY{\mathbb{Y}}
\def\TT{\mathbb{T}}
\def\SS{\mathbb{S}}

\def\prob{\mathsf{P}}
\def\prior{\mathsf{Q}}

\def\credal{\mathscr{Q}}
\def\cred{\mathscr{C}}

\def\lPi{\underline{\Pi}}
\def\uPi{\overline{\Pi}}

\def\joint{\mathsf{R}}
\def\joints{\mathscr{R}}

\def\unif{\mathsf{Unif}}

\title{Valid and efficient possibilistic structure learning in Gaussian linear regression}
\author{Ryan Martin, Naomi Singer, Jonathan P Williams}
\begin{document}

\maketitle

\begin{abstract}
A crucial step in fitting a regression model to data is determining the model's structure, i.e., the subset of explanatory variables to be included.  However, the uncertainty in this step is often overlooked due to a lack of satisfactory methods. Frequentists have no broadly applicable confidence set constructions for a model's structure, and Bayesian posterior credible sets do not achieve the desired finite-sample coverage. In this paper, we propose an extension of the possibility-theoretic inferential model (IM) framework that offers reliable, data-driven uncertainty quantification about the unknown model structure.  This particular extension allows for the inclusion of incomplete prior information about the unknown structure that facilitates regularization.  We prove that this new, regularized, possibilistic IM's uncertainty quantification is suitably calibrated relative to the set of joint distributions compatible with the data-generating process and assumed partial prior knowledge about the structure.  This implies, among other things, that the derived confidence sets for the unknown model structure attain the nominal coverage probability in finite samples. We provide background and guidance on quantifying prior knowledge in this new context and analyze two benchmark data sets, comparing our results to those obtained by existing methods. 

\smallskip

{\it Keywords:} Choquet integral; inferential models, model selection; partial prior; regularization; uncertainty quantification.
\end{abstract}

\section{Introduction} 
\label{sec:intro}

Classical statistical inference focuses on cases where the model structure is given and the goal is quantifying uncertainty about the structure-specific parameters.  Of course, the model structure is rarely given in applications, so a crucial step in modern data analysis is determining the model structure.  In some cases, the structure itself is of primary interest, and the structure-specific parameters themselves are of only secondary importance.  This includes linear regression \citep[e.g.][]{lockhart2014significance, hansen.econometrca.2011, buhlmann.geer.book}, as is the focus in the present paper, where the model structure corresponds to the subset of explanatory variables included.  This also includes the number of components in a finite mixture model \citep[e.g.,][]{chen.khalili.2008, mclachlan1987, mclachlanpeel}, 
and cluster membership in a heterogeneous population \citep[e.g.][]{gan.etal.2021, schafer2016mapping, gao.etal.2023}.  Regardless, the meaningfulness and credibility of inferences drawn, predictions made, etc.~depend crucially on the structure that is chosen.  This structure choice naturally has uncertainty associated with it, but this uncertainty is often overlooked or under-prioritized in the existing literature and across the social and behavioral sciences, largely due to a lack of satisfactory methods for accommodating this uncertainty.  The present paper offers a novel perspective and framework for provably reliable uncertainty quantification about the model structure, thereby offering new statistical techniques for more reliable model/variable selection and post-selection inference, which are standard aspects of inferences from psychological data \citep{bergh.etal.2021,claeskens2008}.

We start here with a general setup.  Let $Y$ denote the observable data taking values in a sample space $\YY$.  Posit a statistical model $\{\prob_\theta: \theta \in \TT\}$, where $\prob_\theta$ is a probability distribution supported on (measurable subsets of) the sample space $\YY$, depending a parameter $\theta$ taking values in the parameter space $\TT$ and perhaps on other things, such as explanatory variables.  The underlying parameter $\theta$ splits into two components, $\theta \to (s,\theta_s)$, where $s$ is the model structure and $\theta_s$ are the structure-specific parameters.  The example we focus on in the present paper is where $Y = (Y_1,\ldots,Y_n)^\top$ is a $n$-vector of response variables, $x$ is a $n \times p$ matrix of explanatory variables, $\theta$ is a triple $(s,\phi_s,\lambda)$ consisting of the subset of columns of $x$ to be included, the regression coefficients associated with those columns, and an error variance; then the associated Gaussian linear regression model is $\prob_\theta = \mathsf{N}_n(x_s\phi_s, \lambda I_{n \times n})$. Throughout, we follow the convention that the true or unknown or uncertain values $Y$, $\Theta = (S,\Theta_S)$, etc.~will be denoted by upper-case letters and specific values $y$, $\theta=(s,\theta_s)$, etc.~will be denoted by the corresponding lower-case letters.  Then the goal is to learn about the uncertain structure $S$, to quantify uncertainty about $S$, given $Y=y$.   

A key point underscoring the  challenges associated with structure learning, compared to inference about structure-specific parameters, is that the data will {\em always} favor the most complex structure.  Therefore, if one has reason to believe that certain simpler structures are more plausible than complex structures, then the conclusions can only reflect this if said beliefs are formally incorporated into the data analysis.  That is, in one way or another, the complex models must be suitably penalized.  In the existing literature, there are two dominant approaches to carrying out this discounting.  The first defines an objective function $s \mapsto \text{fit}_y(s) + \gamma \, \text{pen}(s)$, where $\text{fit}_y(s)$ is a data-dependent measure of the quality of the $s$-specific model fit on observed data $y$, such as the negative log-likelihood $-\log L_y(s, \hat\theta_s)$, $\text{pen}(s)$ is a penalty that discounts models that are more complex, and $\gamma > 0$ is a tuning parameter that controls the strength of the penalty's influence. Common examples of penalties include the various information criteria, e.g., AIC, BIC, DIC, etc.~\citep{akaike1973, schwarz1978, hurvich.tsai.1989, dic.2002}, and, indirectly, penalties on the parameter vector, e.g., lasso \citep{tibshirani1996regression}, ridge \citep{hoerl1970ridge}, and SCAD \citep{fan2001variable}. Minimizing this objective function results in a data-driven choice of model structure that balances fit and complexity. Under certain conditions, the minimizer can consistently estimate $S$ \cite[e.g.][]{bozdogan1987model, ding2018model, shao1993linear}. These approaches are focused on estimating $S$, but what about quantifying uncertainty about $S$?  Some suggest exponentiating the negative of the objective function and normalizing to construct a data-driven probability distribution for $S$, but no guidance is provided on how small such a ``probability that $S=s$'' must be to consider structure $s$ refuted by the data.  

Second, under the Bayesian framework, marginal inference on $S$ is carried out by integration with respect to a prior. Given priors $(\Theta_s \mid S=s) \sim \prior_s$ for each structure $s$, the Bayesian evaluates the marginal likelihoods $m_y(s) = \int L_y(S, \theta_s) \, \prior_s(d\theta_s)$.  One way to use the marginal likelihood is to make pairwise comparisons between two candidate structures, say, $s_1$ and $s_2$ using the Bayes factor, $m_y(s_1)/m_y(s_2)$: if this ratio is large, then conclude that the data favors the model structure $s_1$ over $s_2$. Another way introduces a (marginal) prior distribution $S \sim \prior$ and uses Bayes's theorem to construct a posterior mass function for $S$, i.e., $\prior_y(s) \propto m_y(s) \, \prior(s)$.
This provides probabilistic uncertainty quantification about $S$, but its meaningfulness and properties depend on the priors $\prior_s$ for the $s$-specific parameters. If no prior information exists, the literature generally suggests using default ``flat" priors, but this is not allowed in structure learning \citep[e.g.,][]{berger2001, llorente2022}.  Therefore, the priors $\prior_s$ must be proper and, hence, informative. If no genuine prior information about the structure-specific parameters is available, then the chosen prior distributions are artificial, and their informativeness evidently negatively affects the interpretation and properties of the Bayesian solution. Even if prior information about the structure-specific parameters is available, rarely is it so informative to determine a single (precise) prior distribution from the set of all possible prior distributions. In any case, the Bayesian posterior distributions for $S$, given $y$, come with no calibration guarantees, so, e.g., the corresponding credible sets need not be confidence sets.  To our knowledge, all that has been established is a type of posterior consistency that ensures the credible set contains the true $S$ with probability converging to 1 as $n \to \infty$ \citep[e.g.,][]{casella2009consistency, chae.etal.glm, castillo2015bayesian}. 

For practitioners, this dichotomy between classical/frequentist and Bayesian solutions is unsatisfactory and leads to confusion about when to employ one regime and not the other and how to compare them \citep[e.g.,][]{nelson2020quantifying}. Further, an awkward tension becomes clear: either incorporate no prior information and sacrifice probabilistic uncertainty quantification altogether or introduce proper prior distributions in a probabilistic Bayesian analysis and sacrifice calibration/credibility.
Our view is that a statistical framework should not dictate the assumptions that practitioners make but instead should smoothly accommodate the assumptions that practitioners are willing to make. The partial prior inferential model (IM) framework introduced in \citet{martin2022valid1, martin2022valid2, martin2023valid3} achieves this goal, but has only been developed for the case of uncertainty quantification about structure-specific model parameters, namely, $\Theta_s$. The present paper extends this framework to provide provably reliable uncertainty quantification about the model structure, along with guidance on how prior information about model structure can be incorporated.  

Section~\ref{sec:background} provides necessary background on possibility theory and a summary of IMs and their properties \citep[e.g.,][]{martin2022valid1, martin2022valid2, martin2023valid3} in the case when prior information is vacuous.  Section~\ref{sec:add_ms} begins with formalizing the notion of partial prior knowledge: too much information to ignore, but not enough to justify a full set of prior distributions as Bayesian methods require, e.g., knowing the model structure must be sparse but having no information about which parameters are non-zero or the particular non-zero values. The remainder of Section~\ref{sec:add_ms} discusses how this partial prior information is incorporated into the IM so that reliable uncertainty quantification about model structure is achieved, and provides an empirical demonstration of the IM's validity.  So as not to conflate the partial-prior IM developments with their practical implementations/approximations, computational details are relegated to Appendix~\ref{A:computation}.  As this is a general methodology paper, rather than one focused on a particular application, we cannot suggest any specific choice of partial prior for practitioners to use.  But we can offer some general guidance, and we do so in Section~\ref{sec:prior_contour}.  Section~\ref{sec:rda} applies these guidelines to analyze two benchmark data sets in the literature, and Section~\ref{sec:conclusion} offers some concluding remarks.

\section{Background} 
\label{sec:background}

\subsection{Possibility theory}
\label{SS:possibility}

Possibility theory \citep[e.g.,][]{dubois.prade.book} is one of the simplest imprecise probability models, closely related to fuzzy set theory \citep[e.g.,][]{klir1999fuzzy, zadeh1978fuzzy}, and can be viewed as a special case of Dempster--Shafer theory \citep[e.g.,][]{dempster2008, shafer1976, yagerliu2008}.  The role of possibility theory and possibilistic reasoning in statistical applications is described in \citet{dubois2006} and the references therein.  

A possibility measure supported on a space $\TT$, intended to quantify uncertainty about $\Theta$ in $\TT$, is defined in terms of a real-valued function $\pi: \TT \to [0,1]$, called a {\em possibility contour} with the property that $\sup_{\theta \in \TT} \pi(\theta)=1$.  The latter supremum-equals-1 condition on the contour is the possibility-theoretic analogue of the familiar condition that a probability density function integrates to 1.  This observation reveals the key mathematical difference between possibility and probability theory: optimization is to possibility theory what integration is to probability theory.  Indeed, the {\em possibility measure} $\uPi$ corresponding to the contour $\pi$ is defined as 
\begin{equation}
\label{eq:contour.to.poss}
\uPi(H) = \sup_{\theta \in H} \pi(\theta), \quad H \subseteq \TT. 
\end{equation}
The corresponding {\em necessity measure} $\lPi$, which complements the possibility measure, is defined as $\lPi(H) = 1 - \uPi(H^c)$, but we will not have much need for this below.  

Before we move on to discuss important details about how possibility measures are interpreted and how they are relevant for probabilistic reasoning and uncertainty quantification, we give two quick illustrations where visualization helps to make the theory concrete.  While it is not the only way, one general way to construct a possibility contour is via the so-called {\em probability-to-possibility transform} \citep[e.g.,][]{dubois.etal.2004, hose2020data, hose2021universal}.  This construction starts with a random variable $T$ on $\TT$ that has a probability density/mass function $f$.  Then the probability-to-possibility transform is 
\[ \pi(\theta) = \prob\{ f(T) \leq f(\theta) \}, \quad \theta \in \TT. \]
It is easy to see that $\pi$ is a contour, since it attains value 1 at the mode of $f$.  Figure~\ref{fig:toy} plots the possibility contour obtained by this construction for two cases---$T$ is gamma and $T$ is binomial---and shows computation of $\uPi(H)$ for $H=[2.5,5]$ via optimization.   

\begin{figure}[t]
\begin{center}
\subfigure[$T \sim {\sf Gamma}(2,1)$]{\scalebox{0.6}{\includegraphics{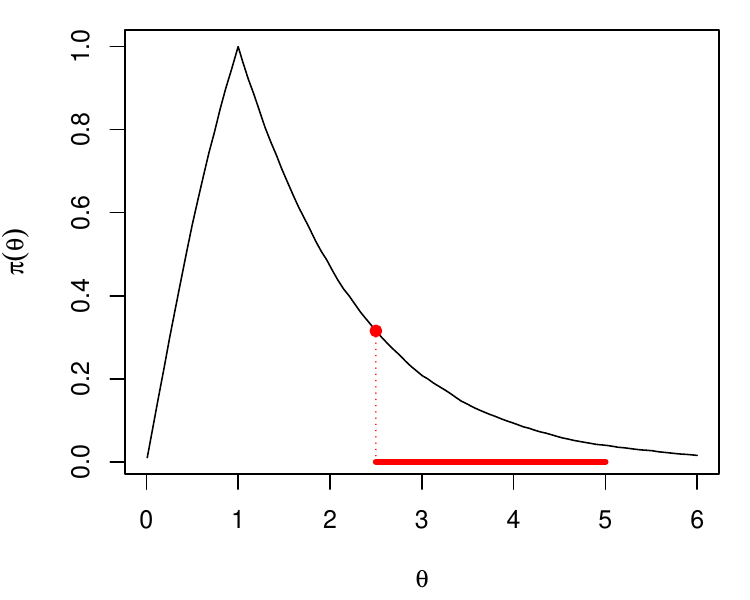}}}
\subfigure[$T \sim {\sf Bin}(6, 0.2)$]{\scalebox{0.6}{\includegraphics{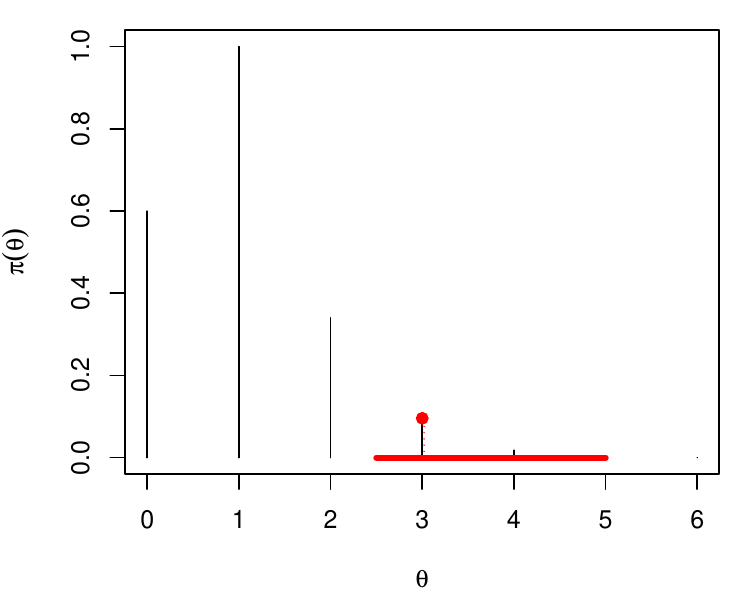}}}
\end{center}
\caption{Plots of two possibility contours based on the probability-to-possibility transform described in the text.  Both plots show the hypothesis $H=[2.5,5]$ in red, with the corresponding $\uPi(H)$ defined by optimization of the contour over $H$.}
\label{fig:toy}
\end{figure}

Since possibility theory is like probability theory but with integration replaced by optimization, it is clear how to carry out marginalization.  That is, if $\Phi = g(\Theta)$ is a given function of the uncertain $\Theta$, then a marginal possibility contour for $\Phi$ is given by 
\begin{equation}
\label{eq:extension0}
\pi^\text{\sc ex}(\phi) = \sup_{\theta: g(\theta) = \phi} \pi(\theta), \quad \phi \in g(\TT). 
\end{equation}
This is an instantiation of the {\em extension principle} coined in \citet{zadeh1975a}, where inference about one unknown is extended to inference about another related unknown; see below.  

Possibility theory is no less rigorous than probability theory; the two are just different and so the former, which is less familiar, requires explanation.  First, note that a possibility contour $\pi$ determines a set of probability distributions, called the {\em credal set} associated with $\pi$ \citep[e.g.,][]{augustin2014introduction, huber1981, walley1991, levi1980}, defined as 
\begin{equation}
\label{eq:credal}
\cred(\pi) = \{\prob \in \text{probs}(\TT): \text{$\prob(H) \leq \uPi(H)$ for all measurable $H$}\}, 
\end{equation}
where $\text{probs}(\TT)$ is the collection of all probability measures supported on (a $\sigma$-algebra of subsets of) $\TT$.  The supremum-equals-1 normalization ensures that this set is non-empty and, therefore, that $\uPi$ can be interpreted as a coherent upper probability \citep[e.g.,][Ch.~7]{lower.previsions.book}.  This is how possibility measures are interpreted and how they facilitate probabilistic reasoning: if $\uPi(H)$ is small, then $\prob(H)$ is likewise small for all $\prob$ in the credal set $\cred(\pi)$.  It is worth mentioning that the credal set determined by $\pi$ has other convenient and insightful characterizations.  For instance, \citet{cuoso.etal.2001} show that the random variable $\pi(\Theta)$ behaves like a p-value with respect to all $\prob$ in $\cred(\pi)$, i.e., 
\[ \sup_{\prob \in \cred(\pi)} \prob\{ \pi(\Theta) \leq \alpha \} \leq \alpha, \quad \alpha \in [0,1]. \]

We end this review with two notable things that possibility theory can do that probability theory cannot.  The first is that possibility theory can easily model ignorance: just take $\pi(\theta) \equiv 1$.  The result is $\uPi(H) = 1$ for all (non-empty) $H \subseteq \TT$ and, therefore, the credal set $\cred(\pi)$ contains all the probability distributions supported on $\TT$; that is, if we are ignorant about $\Theta$, then no probability distribution can be excluded from consideration when quantifying uncertainty about it.  It is obvious that no single probability distribution, however diffuse, can do the job of the set of all probability distributions; in particular, complete ignorance includes non-diffuse distributions.  The second, which is related to the first, is that we can literally extend our possibilistic uncertainty quantification about $\Theta$ to other quantities that are not fully determined by $\Theta$.  Suppose, for example, that $\Phi = (\Theta, \Lambda)$ and we are fully ignorant about $\Lambda$ and how it might be related to $\Theta$.  Then we can quantify uncertainty about $\Phi$ using the extended possibility contour 
\[ \pi^+(\theta, \lambda) = \pi(\theta), \quad \text{for all $(\theta,\lambda)$}. \]
In this way, if we marginalize to $\Theta$ according to the rule \eqref{eq:extension0}, then we recover $\pi(\theta)$ and, similarly, if we marginalize to $\Lambda$, then we get a constant contour corresponding to ignorance.  The point is that this instantiation of Zadeh's extension principle preserves what is known about $\Theta$ and $\Lambda$ individually, but is ignorant concerning their relationship.

\subsection{Possibilistic IMs}
\label{SS:vacuous.im}

\subsubsection{Motivation}

The IM framework, first developed in \citet{imbasics, imbook}, aimed to fill an awkward gap between the Bayesian and frequentist mainstreams where, on the one hand, probabilism is paramount and, on the other hand, sampling-based calibration is cardinal.  An early conjecture was that probabilism and calibration are fundamentally incompatible, later confirmed by the {\em false confidence theorem} \citep{balch2019satellite}: there is no construction of a data-dependent probability distribution quantifying uncertainty about unknown model parameters that is provably calibrated across all hypotheses; see, also, \citet{martin.nonadditive, martin.belief2024} and the concrete example in Section~\ref{SSS:toy.example} below.  This explains why Fisher's fiducial program fell short and why default-prior Bayes is not fully satisfactory.  That the root cause of false confidence is additivity motivates consideration of data-dependent {\em imprecise probabilities} for uncertainty quantification in statistical applications.  

To our knowledge, the original IM formulation---``one of the original statistical innovations of the 2010s'' \citep{cui.hannig.im}---was the first to be provably safe from false confidence.  That formulation focused on imprecise-probabilistic uncertainty quantification modeled via random sets and belief functions, but there are certain challenges one faces in putting this approach into practice.  A different and arguably simpler formulation emerged in \citet{plausfn, gim}, but only recently did the possibility-theoretic foundations behind it start to take shape.  Our formulation below closely follows that in  \cite{martin2022valid2}.

\subsubsection{Construction}

Let $\{\prob_\theta: \theta \in \TT\}$ be a posited statistical model for the observable data $Y$.  In particular, assume that $Y \sim \prob_\Theta$, where $\Theta \in \TT$ represents the true but unknown/uncertain value of the parameter.  For the moment, we assume that no prior information relevant to $\Theta$ is available, i.e., we assume {\em a priori} ignorance about $\Theta$.  Write $\theta \mapsto L_y(\theta)$ for the likelihood function.  The goal is to reliably quantify uncertainty about the unknown $\Theta$ using only the observed data and information contained in the model/likelihood.  

The possibilistic IM formulation in \citet{martin2022valid2} proceeds in two steps: {\em ranking} and {\em validification}.  The ranking step amounts to specifying a $y$-dependent ranking or preference ordering on the candidate $\theta$ values that takes into account information about the model.  One natural choice of ranking is that based on the relative likelihood 
\begin{equation}
\label{eq:rellik}
\eta(y,\theta) = \frac{L_y(\theta)}{\sup_\vartheta L_y(\vartheta)}, \quad \theta \in \TT.
\end{equation}
This relative likelihood-based ranking is the one that we will adopt here, modulo some simple but important profiling and penalization adjustments.  Next, the validification step proceeds by applying a variation on the probability-to-possibility transform described earlier:
\begin{equation}
\label{eq:contour0}
\pi_y(\theta) = \prob_\theta\{ \eta(Y,\theta) \leq \eta(y,\theta) \}, \quad \theta \in \TT,
\end{equation}
where the probability $\prob_\theta$ determines the distribution of $Y$ that appears in the relative likelihood function $\eta(Y,\theta)$.  This resembles the p-value corresponding to a likelihood ratio test, which is not surprising given the close connection between possibility theory and p-values.  Since the model parameters are often continuous-valued, plots of the possibilistic IM's contour function $\theta \mapsto \pi_y(\theta)$ tend to look similar to that in Figure~\ref{fig:toy}(a).  Given the contour $\pi_y$, the possibilistic IM construction is completed by defining the corresponding possibility measure $\uPi_y$ via optimization; $\uPi_y(H) = \sup_{\theta \in H} \pi_y(\theta)$, just as in \eqref{eq:contour.to.poss}.  \citet{martin2022valid1} showed that the members of the possibilistic IM's credal set $\cred(\pi_y)$ have a {\em confidence distribution} interpretation.  So, if $\uPi_y(H)$ is small, then all of the confidence distributions assign low credence to the truthfulness of the hypothesis ``$\Theta \in H$,'' and this is what makes the IM's output meaningful.  Such sampling-based reliability properties are discussed in Section~\ref{SSS:im.vac.properties}. 

It is almost always the case that interest is in some feature $\Phi=g(\Theta)$ of the full parameter $\Theta$.  There are (at least) two ways one could construct a marginal IM for $\Phi$.  The first is a direct application of the extension principle \eqref{eq:extension0} discussed above; write this as $\phi \mapsto \pi_y^\text{\sc ex}(\phi)$.  The other tailors the IM construction specifically to the quantity $\Phi$ of interest by recasting the ranking function as a relative profile likelihood 
\begin{equation}
\label{eq:relproflik}
\eta^\text{\sc pr}(y,\phi) = \frac{\sup_{\theta: g(\theta)=\phi} L_y(\theta)}{\sup_{\theta \in \TT} L_y(\theta)}, \quad \phi \in g(\TT). 
\end{equation}
Then carry out the validification step similar to above:
\[ \pi_y^\text{\sc pr}(\phi) = \sup_{\theta: g(\theta) = \phi} \prob_\theta\{ \eta^\text{\sc pr}(Y,\phi) \leq \eta^\text{\sc pr}(y,\phi) \}, \quad \phi \in g(\TT). \]
The two functions, $\pi_y^\text{\sc ex}$ and $\pi_y^\text{\sc pr}$, are different and one generally cannot conclude that one is ``better'' than the other in a given application, but empirical evidence suggests \citep{martin2023valid3} and theoretical results confirm \citep{imbvm.ext} that the simpler extension-based marginalization is generally less efficient than profile-based marginalization.

\subsubsection{Properties}
\label{SSS:im.vac.properties}

The possibilistic IM's key property is a simple one: 
\begin{equation}
\label{eq:valid0}
\sup_{\theta \in \TT} \prob_\theta\{ \pi_Y(\theta) \leq \alpha \} \leq \alpha, \quad \alpha \in [0,1].
\end{equation}
Again, this closely resembles the familiar property satisfied by p-values, but we are considering here a much broader scope.  The point is that IM's possibility contour function evaluated at the true value $\Theta$, no matter what the true value happens to be, is calibrated in the sense that $\pi_Y(\Theta)$ is stochastically no smaller than a $\unif(0,1)$ random variable.  The cardinal calibration property \eqref{eq:valid0} has several important consequences.  First, if we define 
\[ C_\alpha(y) = \{ \theta \in \TT: \pi_y(\theta) > \alpha\}, \quad \alpha \in [0,1], \]
then $C_\alpha(Y)$ is a $100(1-\alpha)$\% frequentist confidence set in the sense that 
\[ \sup_{\theta \in \TT} \prob_\theta\{ C_\alpha(Y) \not\ni \theta\} \leq \alpha, \quad \alpha \in [0,1]. \]
Second, if we take any fixed $H \subseteq \TT$, then 
\begin{equation}
\label{eq:valid.alt0}
\sup_{\theta \in H} \prob_\theta\{ \uPi_Y(H) \leq \alpha \} \leq \alpha, \quad \alpha \in [0,1]. 
\end{equation}
That is, if $H$ is a true hypothesis in the sense that $H \ni \Theta$, then it is a $\prob_\Theta$-rare event that the IM assigns small upper probability to $H$, and ``rare'' and ``small'' are explicitly linked.  The result in \eqref{eq:valid.alt0} can be rewritten in terms of the IM's lower probability or necessity as follows:
\[ \sup_{\theta \not\in H} \prob_\theta\{ \lPi_Y(H) \geq 1-\alpha \} \leq \alpha, \quad \alpha \in [0,1]. \]
If ``false confidence'' means large support for false hypotheses, then this implies that the IM is able to control its false confidence rate whereas data-dependent probabilities generally cannot.  It turns out that property \eqref{eq:valid.alt0} not only holds for fixed $H$, but also {\em uniformly} in $H$; see \citet{cella.martin.probing} and Section~\ref{SS:valid} below.  

The kind of validity properties described above are easy to achieve simply by being overly cautious.  But despite the inherent imprecision in the possibilistic IM's output, caution is not the reason why validity is satisfied.  This is obvious when the framework is applied to a real problem (see Section~\ref{SSS:toy.example}), but this can also be proved theoretically for large samples under the usual regularity conditions.  Indeed, \citet{imbvm.ext} prove a possibility-theoretic version of the celebrated Bernstein--von Mises theorem; the one that establishes that Bayesian \citep[e.g.,][Ch.~10.2]{vaart1998} and generalized fiducial inference \citep[e.g.,][]{hannig.review} are asymptotically valid and efficient, where efficiency means that the ``posterior variance'' matches the Cram\'er--Rao lower bound.  Interested readers are referred to \citet{imbvm.ext} for details and to \citet{immc, reimagined} for some related developments.

\subsubsection{Example}
\label{SSS:toy.example}

Consider the simple (fixed structure) Gaussian linear regression model 
\begin{equation} \label{eq:lm}
    (Y \mid x, \theta) \sim \prob_{\phi,\lambda} := {\sf N}_n(x\phi, \lambda I_{n \times n}), \quad \theta = (\phi, \lambda),
\end{equation}
where $Y$ is an $n$-vector of response variables, $x$ is an $n \times (p + 1)$ matrix whose first column is filled with ones and the remaining columns correspond to the measured explanatory variables, $\phi^\top = (\phi_0, \phi_1^\top)$ is a $(p+1)$-vector of regression coefficients where $\phi_0$ corresponds to the intercept and $\phi_1^\top$ corresponds to the $p$ parameters associated with the covariates, and the scalar $\lambda$ is an error variance.  It is easy to check that, modulo irrelevant constants, the log-relative likelihood is given by 
\[ \log \eta(y,\phi,\lambda) = n\log(\hat\lambda_y / \lambda) - \|y - x\phi\|_2^2 / \lambda, \]
where $\hat\lambda_y = n^{-1} \|y - x \hat\phi_y\|_2^2$ and $\hat\phi_y = (x^\top x)^{-1} x^\top y$ are the maximum likelihood estimators of $\lambda$ and $\phi$, respectively.  Since $\eta(Y,\phi,\lambda)$ is a pivot with respect to $\prob_{\phi,\lambda}$, it is easy to carry out the validification step (numerically) to get the possibilistic IM contour $\pi_y(\phi,\lambda)$ for each $(\phi,\lambda)$; there is, unfortunately, no simple closed-form expression to show here.  The more common situation in applications, however, is where the error variance is a nuisance parameter and the focus is solely on the unknown regression coefficients, $\Phi$.  Then we would get a direct marginal IM for $\Phi$ via profiling.  Indeed, the log-relative profile likelihood is 
\[ \log\eta^\text{\sc pr}(y,\phi) = -n \log\{ \hat\lambda_y(\phi) / \hat\lambda_y \}, \]
where $\hat\lambda_y$ is as before, and $\hat\lambda_y(\phi) = n^{-1} \|y - x\phi\|_2^2$ is the maximum likelihood estimator of the error variance when $\Phi=\phi$ is assumed known.  Accordingly, the profile-based marginal IM for $\Phi$ has possibility contour 
\[
\pi_y^\text{\sc pr}(\phi) = \sup_\lambda \prob_{\phi,\lambda} \Bigl\{ \frac{\hat\lambda_Y(\phi)}{\hat\lambda_Y} \geq \frac{\hat\lambda_y(\phi)}{\hat\lambda_y} \Bigr\} = 1 - F\Bigl( \frac{\|x\hat\phi_y - x\phi\|_2^2}{\|y - x\hat\phi_y\|_2^2} \Bigr), \quad \phi \in \mathbb{R}^{p+1},
\]
where $F$ is the ${\sf F}_{p+1, n-p-1}$ distribution function.  Consequently, the level sets are 
\[ C_\alpha(y) = \{ \phi \in \mathbb{R}^{p+1}: \pi_y^\text{\sc pr}(\phi) > \alpha\} = \Bigl\{ \phi \in \mathbb{R}^{p+1}: \frac{\|x\hat\phi_Y - x\phi\|_2^2}{\|Y - x\hat\phi_Y\|_2^2} < F^{-1}(1-\alpha) \Bigr\}, \]
which agrees with the textbook $100(1-\alpha)$\% confidence set for the regression coefficients.  

That there are similarities between the (marginal) possibilistic IM solution and the textbook frequentist solution is not a shortcoming of the former; in fact, the latter is the ``textbook solution'' because it is the best in a certain sense, so it would be more concerning if the former differed from the latter.  This highlights that the asymptotic efficiency enjoyed by the possibilistic IM in general is, in some cases at least, exact efficiency.  It is important to keep in mind that the IM solution is more than just a confidence set---it offers fully conditional possibilistic uncertainty quantification that is provably reliable without sacrificing efficiency.  As a quick demonstration that this is nontrivial, consider a standard, default-prior Bayes solution.  In this case, the posterior distribution $\prior_y^\text{\sc bayes}$ is the standard normal--inverse gamma, which is easy to simulate, so posterior inference is rather straightforward.  To see how false confidence can creep in, consider the case where $p=1$.
A relevant quantity might be the ratio $-\Phi_0/\Phi_1$, which represents the root of the mean response function; this is like ``inverse regression.''  Suppose the true coefficient values are $\Phi_0=0.3$ and $\Phi_1=0.1$, with error variance $\lambda=1$, so that the true value of the root is $-3$, and then the hypothesis $H = \{\phi: -\phi_0 / \phi_1 > -1\}$ is false. Figure~\ref{fig:fc} shows a plot of the empirical distribution function of the random variable $\prior_Y^\text{\sc bayes}(H)$ based on 1000 data sets $Y$ of size $n=25$.  Notice that this distribution function is slowly increasing at first and more rapidly towards the end.  This means that the sampling distribution of $\prior_Y^\text{\sc bayes}(H)$ is concentrated around relatively large values, e.g., $\prior_Y^\text{\sc bayes}(H)$ is bigger than 0.6 about 90\% of the time.  But remember that this hypothesis $H$ is {\em false}, so this tendency to assign high posterior probability---or confidence---to a false hypothesis creates a risk of systematically misleading conclusions.  

\begin{figure}[t]
\begin{center}
\scalebox{0.7}{\includegraphics{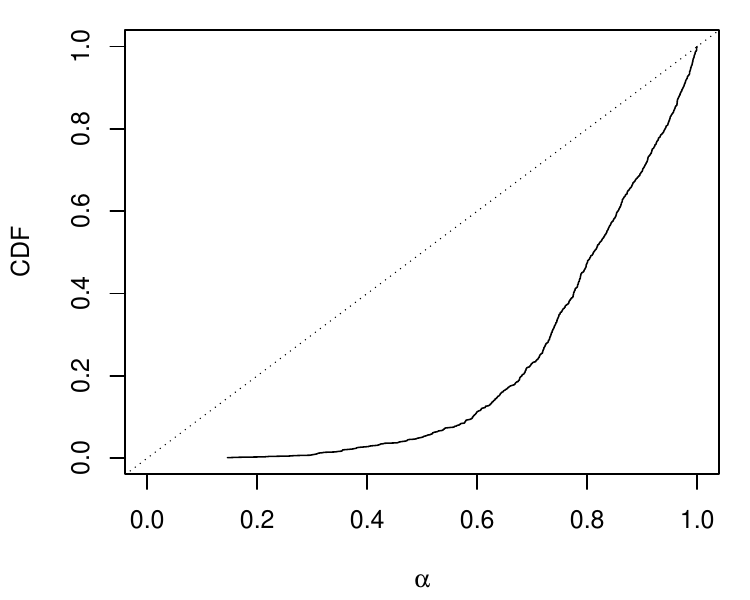}}
\end{center}
\caption{The (empirical) distribution function $\alpha \mapsto \prob\{ \prior_Y^\text{\sc bayes}(H) \leq \alpha\}$ of the Bayes posterior probability assigned to the {\em false} hypothesis $H$ defined in the main text.}
\label{fig:fc}
\end{figure}

\section{Possibilistic IM structure learning} \label{sec:add_ms}

\subsection{Partial prior information}
\label{sec:ppim_constr}

Section~\ref{sec:intro} reviewed the mainstream approaches to structure learning.  What these approaches have in common is that they penalize against overly complex model structures.  But what is the rationale behind penalizing complex structures?  It must be that there is an underlying belief---either from genuine knowledge about the problem at hand or by appeal to general parsimony principles like Occam's razor---that the unknown true model structure is not complex.  Either way, this is {\em prior information} and it is both fair and advantageous to treat it as such.  But this perspective creates an awkward tension: the only mainstream approach available that can incorporate prior information is Bayesian, and adopting a Bayesian perspective obligates the data analyst to introduce prior distributions for all of the unknowns, not just for the model structure which is believed to be not-complex.  Indeed, Manski's law of decreasing credibility \citep{manski2003partial} says that the credibility of inference decreases with the strength of assumptions made, so introducing (necessarily ``informative'') prior distributions for the structure-specific parameters amounts to assuming more than can be justified, hence the inferences lose credibility.  Our proposal offers a sort of middle ground, one where only the available-but-vague prior information about the model structure can be incorporated into the analysis, and, by interpreting this addition as part of the statistical model itself, the data analyst can enjoy both the efficiency gains afforded by penalizing overly complex models and the desirable frequentist-like reliability guarantees in structure learning.  

Our starting point is the observation that statements or assumptions, e.g., ``the model structure is not too complex,'' amount to incomplete prior information about the full parameter $\Theta$.  As a simple example, suppose, like in the classical cases mentioned at the beginning of Section~\ref{sec:intro}, that the model structure $S$ is {\em known} to be, say, $s$.  This says absolutely nothing about the structure-specific parameter $\Theta_s$ so our prior information in that direction is fully vacuous.  That is, our {\em a priori} uncertainty about $\Theta_s$ would be quantified by saying that the ``prior probability that $\Theta_s \in H$ is in $[0,1]$'' for any $H \subset \TT_s$. This knowledge state can be described mathematically via a credal set that contains {\em all prior distributions} for $\Theta_s$.  

Under our present assumption that ``$S$ is not too complex,'' partial or incomplete prior information about $S$ determines a set $\credal$ of prior distributions for $\Theta=(S,\Theta_S)$.  The corresponding credal set $\credal$ consists of mixtures of the form 
$ \sum_{s \in \SS} \prior_S(s) \, \prior_{\Theta_s}(\cdot)$, 
where $\prior_{\Theta_s}$ can be any prior distribution for $\Theta_s$ supported on $\TT_s$ and the weights $\prior_S(s)$ assigned to the different model structures must be compatible with the available [partial] prior information about the uncertain model structure $S$.  Further, from a possibilistic construction, all that needs to be specified is a possibility contour function $q$ for $\Theta$, which is much simpler than directly specifying a credal set.  Of course, there is still a credal set, i.e., $\credal=\cred(q)$ as defined in \eqref{eq:credal}, but this admits a special structure that helps simplify the relevant computations below.  

For reasons that will be made clear below, the choice of prior contour $q$---or, more generally, the prior credal set $\credal$---should be treated as part of the model-building process and, in that sense, it is up to the user to determine what best suits his/her application.  Some guidance on how the partial prior assessment can be made is given in Section~\ref{sec:prior_contour} below but, for the remainder of this section, we will take the partial prior contour $q$ as given, focusing for now on how $q$ can be used and what properties can be established if used in this way.  All that will be assumed about $q$ below is that $q(s,\theta_{s})$ only depends on the structure $S=s$ inherent in the given parameter $\theta$; i.e., $q$ is constant on $\theta$'s having a common structure. 

There is one last technical detail to present before getting into the partial-prior IM construction.  In a Bayesian analysis, when one posits a prior distribution for the parameter, this, together with the statistical model---interpreted as a data-given-parameter conditional distribution---determines a joint distribution for the data--parameter pair $(Y,\Theta)$.  Likewise, when one posits a credal set $\credal$ of prior distributions, this determines a set of joint distributions for $(Y,\Theta)$.  Write $\joint_\prior$ for the joint distribution determined by a prior $\prior \in \credal$, and define the set of such joint distributions as $\joints = \{\joint_\prior: \prior \in \credal\}$.  So, if $E \subseteq \YY \times \TT$ is an event concerning the pair $(Y,\Theta)$, then the upper probability of $E$ is defined as 
\begin{equation}
\label{eq:upper.joint}
\overline\joint(E) = \sup_{\prior \in \credal} \joint_\prior(E) = \sup_{\prior \in \credal} \iint 1\{ (y,\theta) \in E \} \, \prob_\theta(dy) \, \prior(d\theta).
\end{equation}
We will not need to work directly with any of these joint distributions or the above supremum over $\credal$; this is just notation to help with the explanations that follow.  Specifically, we show below that {\em all relevant computations can be carried out using the prior contour $q$}.  This simplicity is one of the benefits of the possibilistic formulation!

\subsection{Partial-prior IM construction}

Here we follow the partial-prior IM construction presented in \citet{martin2022valid2}.  For the sake of completeness, we walk through all the steps below but focus on the case where the model structure $S$ is of primary interest.  Just as in the vacuous-prior case reviewed in Section~\ref{SS:vacuous.im}, we proceed with the two basic steps: ranking and validification.  For the first, analogous to the relative profile likelihood used in \eqref{eq:relproflik}, here we propose a penalized version thereof, where the marginalization is over the structure-specific parameters and the penalty is $q(s,\theta_s)$:
\[ \eta_q(y,s) = \frac{\max_{\theta_s} L_y(s,\theta_s) \, q(s,\theta_s)}{\max_{r,\theta_r} L_y(r,\theta_r) \, q(r, \theta_r)}, \quad s \in \SS. \]
It would often be the case that one has vacuous prior knowledge about the $s$-specific structure parameters, as we are assuming here; in such cases, $q(s,\theta_s)$ is constant in $\theta_s$ for a given $s$ and we can write the contour as $q(s)$.  Then the relative penalized profile likelihood simplifies to 
\[ \eta_q(y,s) = \frac{L_y(s,\hat\theta_s) \, q(s)}{\max_{r \in \SS} L_y(r,\hat\theta_r) \, q(r)}, \quad s \in \SS, \]
where $\hat\theta_s$ is the maximum likelihood estimator of the $s$-specific model parameter.  In the Gaussian linear regression case that we have in mind here, the maximum likelihood terms are available in closed form: modulo irrelevant constants, this is 
\[ L_y(s,\hat\theta_s) = \bigl( \hat\lambda_{y,s}^2 \bigr)^{-n/2}, \quad s \in \SS, \]
where $\hat\lambda_{y,s}^2 = n^{-1} \|y - (x_s^\top x_s)^{-1} x_s^\top y\|^2$ is a scaled residual sum of squares associated with the $s$-specific model.  The presence of the $q$ penalty terms and the fact that the denominator maximizes over model structures breaks the aforementioned connection with the F distribution in the fixed-$s$ case.  But the closed-form expressions for both the likelihood and penalty terms means that we can, at least for low- or moderate-dimension problems, easily evaluate this relative penalized profile likelihood using, e.g., the {\tt leaps} package in R \citep{leaps}.  

The relative penalized profile likelihood is surely a reasonable ranking of the model structures and, especially in the latter case above, there is a clear parallel to the balance-model-fit-and-complexity perspective described in Section~\ref{sec:intro}.  This provides a ranking equivalent to that based on the following objective function: 
$
s \mapsto -\log L_y(s,\hat\theta_s) - \log q(s). 
$
Naturally, then, the ``best'' model structure would be
\begin{equation}
\label{eq:s.hat}
\hat s_y = \arg\min_s \bigl\{ -\log L_y(s,\hat\theta_s) - \log q(s) \bigr\} = \arg\max_s \bigl\{ L_y(s,\hat\theta_s) \, q(s) \bigr\}. 
\end{equation}
The specific functional form that $q$ might take will be discussed further in Section~\ref{sec:prior_contour} but a perfectly reasonable choice is $q(s) = e^{-\gamma|s|}$, where $|s|$ is the complexity of structure $s$ and $\gamma > 0$ is a hyperparameter.  With a choice like this, there are clear connections to the various information criterion-based penalties that are common in the literature.  

To complete the partial-prior IM construction, we next carry out the validification step.  This too is perfectly analogous to the validification step presented in Section~\ref{SS:vacuous.im}, but the mathematical details are less familiar here so it might be more difficult to see at first glance.  Recalling that the mapping $S(\theta)$ reads off the structure inherent in the generic parameter $\theta$, the (marginal) partial-prior IM for $S$, given $Y=y$, has contour function $\pi_y(s)$ defined as 
\begin{equation}
\label{eq:pi.s}
\pi_y(s) = \overline\joint\bigl\{ \eta_q(Y,S(\Theta)) \leq \eta_q(y,s) \bigr\}, \quad s \in \SS. 
\end{equation}
This corresponds to the general formula \eqref{eq:upper.joint} above with event $E=E_{y,s}$ given by 
\[ E_{y,s} = \{(\tilde y, \tilde \theta) \in \YY \times \TT: \eta_q(\tilde y, S(\tilde\theta)) \leq \eta_q(y,s) \}. \]
That \eqref{eq:pi.s} defines a bona fide possibility contour can be seen by verifying that $\pi_y(\hat s_y) = 1$ for each $y$, where $\hat s_y$ is as in \eqref{eq:s.hat}.  This contour defines the IM's possibility measure
\[ \uPi_y(H) = \max_{s \in H} \pi_y(s), \quad H \subseteq \SS. \]
There is a corresponding lower probability $\lPi_y$, but it will not be needed in what follows.

\subsection{Theoretical reliability properties}
\label{SS:valid}

The thesis of this paper is that the partial-prior IM offers a ``best of both worlds'' in the following sense: by incorporating exactly what is known about the model structure in the ranking and validification steps, without embellishment, we can simultaneously achieve Bayesian-like, fully conditional, possibilistic uncertainty quantification and frequentist-like unconditional reliability guarantees under realistic conditions.  The previous subsection showed how to obtain the fully conditional possibilistic uncertainty quantification about $S$, and the goal of the present subsection is to establish the latter claims about frequentist-like reliability.  

The type of reliability proved below differs from the classical frequentist reliability properties in important and relevant ways which deserve explanation.  To start, consider the (unrealistic) case with a single prior distribution that is believed to be ``correct.''  Then Bayesian inference based on this ``correct'' prior is reliable in any relevant sense---this is because the Bayes rule used to define the posterior distribution and the probability theory used to evaluate the posterior distribution's properties are closely tied and fully consistent.  For example, under the posited data--parameter joint distribution, the expected value of a posterior probability is the prior probability, so it follows easily from Markov's inequality that the posterior probability exceeding a large multiple of the presumed ``correct'' prior probability is a rare event.  More generally, suppose that one has a credal set of prior distributions, one of which is believed to be ``correct'' in the sense above.  Then the more conservative {\em generalized Bayes} solution championed by \citet{walley1991} and others shares the same reliability properties just described for the single-prior case; though, having a ``correct'' prior credal set is a weaker condition than having a single ``correct'' prior, and so the inferences are more conservative and hence also weaker.  The conservatism of generalized Bayes is beyond our present scope, but the interested reader can consult \citet{kyburg1987}, \citet{walley2002}, and \citet{gong.meng.update}.  Our goal is valid and efficient uncertainty quantification about $S$, so we opt for a solution that is less conservative than generalized Bayes.  The results we prove below are (generalized) Bayesian in spirit, i.e., the partial prior information is incorporated both in the IM solution and in the metric used to evaluate the reliability, but since we are not doing conditioning, etc.~according to the rules of probability, this reliability is not automatic like it is for (generalized) Bayes and needs to be verified directly.  

We start with the main result.  At a superficial level, this result establishes a property similar to that satisfied by p-values.  But the reader should keep in mind that the present context---involving partial prior information, etc.---is entirely different from that in which p-values are shown to be stochastically no smaller than uniform.  

\begin{theorem}
\label{thm:valid}
Let $\pi_y$ be the marginal possibility contour for $S$ associated with the partial-prior IM, depending implicitly on the prior contour $q$ that quantifies our {\em a priori} uncertainty about the full model parameter $\Theta$.  Let $\overline\joint$ denote the upper joint distribution for $(Y,\Theta)$ in \eqref{eq:upper.joint} relative to the posited model for $Y$, given $\Theta$, and the credal set $\credal=\cred(q)$ determined by the partial prior information.  Then 
\begin{equation}
\label{eq:valid}
\overline\joint\{ \pi_Y(S(\Theta)) \leq \alpha \} \leq \alpha, \quad \alpha \in [0,1]. 
\end{equation}
\end{theorem}

\begin{proof}
See Appendix~\ref{A:properties}. 
\end{proof}

Theorem~\ref{thm:valid} has (at least) two important consequences.  The first concerns the upper probability assigned to ``true hypotheses'' about the uncertain $S$.  A relevant class of hypotheses are those that make claims about the complexity of the uncertain $S$, e.g., $H = \{s: |s| > c\}$ or $H = \{s: |s| \leq c\}$, where $c > 0$ is a specified threshold.  Applied to hypotheses like these, the following corollary implies that the IM will not systematically under- or over-fit.  Specifically, Corollary~\ref{cor:valid} says that ``$\uPi_Y(H)$ is small and $S(\Theta) \in H$'' is a rare event (with respect to all the joint distributions having prior compatible with $q$). So, if we apply the standard inductive logic, then the rule ``judge hypothesis $H$ to be false if $\uPi_Y(H)$ is sufficiently small'' is reliable in the sense that it is a rare event that this inference is wrong.  As such, this result can be compared to the classical frequentist Type~I error rate control. 

\begin{cor}
\label{cor:valid}
Under the setup of Theorem~\ref{thm:valid}, for any hypothesis $H$ concerning $S(\Theta)$, the marginal possibilistic IM satisfies 
\[ \overline\joint\{ \uPi_Y(H) \leq \alpha \text{ and } S(\Theta) \in H \} \leq \alpha, \quad \alpha \in [0,1]. \]
\end{cor}

\begin{proof}
This result is implied by the stronger one in Corollary~\ref{cor:uvalid} below.
\end{proof}

Importantly, the above validity property not only holds for any fixed hypothesis $H$, it holds uniformly in hypotheses, revealing how strong the original calibration property in Theorem~\ref{thm:valid} really is.  For example, the IM is safe from false confidence even if the hypotheses are allowed to be data-dependent in some way.  

\begin{cor}
\label{cor:uvalid}
Under the setup of Theorem~\ref{thm:valid}, the marginal possibilistic IM satisfies 
\[ \overline\joint\{ \text{$\uPi_Y(H) \leq \alpha$ for some $H$ such that $S(\Theta) \in H$} \} \leq \alpha, \quad \alpha \in [0,1]. \]
\end{cor}

\begin{proof}
See Appendix~\ref{A:properties}. 
\end{proof}

The next property is more familiar than the previous one, closely resembling the classical coverage probability bounds for confidence sets.  Here, if $\pi_y$ is the marginal, partial-prior IM contour for $S$, then define the upper level sets 
\begin{equation}
\label{eq:level.set}
C_\alpha(y) = \{s: \pi_y(s) > \alpha\}, \quad \alpha \in [0,1]. 
\end{equation}
Corollary~\ref{cor:coverage} below says that $C_\alpha(Y)$ is a $100(1-\alpha)$\% confidence set for $S$, in the sense that the $\overline\joint$-probability that $C_\alpha(Y)$ misses the true structure is no larger than $\alpha$.  But note that this is different from the familiar coverage probability bounds for confidence intervals in how it depends on the partial prior information.  In particular, the ``probability'' used to evaluate the (non-)coverage probability depends on $q$ or $\credal$; rather than surgically removing overly complex structures from the usual coverage probability calculation, we effectively average over those structures compatible with the (imprecise) prior information.  

\begin{cor}
\label{cor:coverage}
Under the setup of Theorem~\ref{thm:valid}, 
\begin{equation}
\label{eq:coverage}
\overline\joint\{ C_\alpha(Y) \not\ni S(\Theta) \} \leq \alpha, \quad \alpha \in [0,1]. 
\end{equation}
\end{cor}

\begin{proof}
Since $C_\alpha(Y) \not\ni S(\Theta)$ is equivalent to $\pi_Y(S(\Theta)) \leq \alpha$, just apply \eqref{eq:valid}.  
\end{proof}

\subsection{Empirical reliability demonstration}
\label{SS:demo}

The goal of the present section is to empirically confirm the coverage probability property advertised in Corollary~\ref{cor:coverage}. It is important to note that this property is with respect to the imprecise joint distribution which depends on the credal set $\credal$ and, thus, on $q$. However, the relationship between $\joint_\prior$ and $\overline\joint$ requires that coverage is also satisfied by any joint distribution $\joint_\prior$ in the credal set induced by $q$. Therefore, it suffices to sample from a precise prior probability distribution consistent with $q$.

For our case, since the prior knowledge about $(S,\Theta_S)$ is vacuous given $K=|S|$, we only need to know how to sample $K$. Once $K=|S|$ is given, the values of $(S,\Theta_S)$ can be chosen in literally any way, e.g. fixed at any value, sampled from any distribution, etc. However, care must be taken in sampling $K$ in order for the prior distribution to be compatible with the contour $q_K$ for $K$. Specifically, since the upper probability determined by $q_K$ must dominate any prior $\prior_K$ for $K$, and since the complexity $K$ is a discrete variable, we must choose $\prior_K$ such that $\prior_K(K=k) \leq q_K(k)$ for each $k$.  For our simple illustration here, we opt to take $\prior_K$ to be the ``most diffuse'' of those probability distributions compatible with $q_K$.  Since $q_K(k) = 0$ for $k > p$, we propose to take the prior probability distribution $\prior_K$ such that 
\[ \prior_K(K \geq k) = q_K(k), \quad k=0,1,\ldots,p. \]
In particular, if $q_K(k)$ is decreasing in $k$, which is natural when the goal is to penalize model structures that are too complex, then the probability mass function must be 
\begin{align} 
\prior_K(K=k) & = q_K(k) - q_K(k+1), \label{eq:Q_K} 
\qquad k=0,1,\ldots,p.
\end{align}
In what follows, we take the prior contour to be $q_K(k) = e^{-\gamma k}$ for $\gamma > 0$ to be specified. 

To demonstrate that our method achieves finite-sample validity, we generate 10,000 data sets from equation \eqref{eq:lm} and determine whether the true structure $S$ is contained in the confidence set constructed using the upper level sets of equation \eqref{eq:level.set}.  More specifically, we start by generating the covariates in $x$ as independent standard normals.  This is done once at the outset and treated as fixed throughout the simulation.  The dimension is $p=3$ and $n=25$, just to keep the details manageable.  To sample data in a way that is compatible with the partial prior information, we proceed as follows.  Start by drawing $K$ from the prior distribution $\prior_K$ in \eqref{eq:Q_K} above, with $\gamma=1$.  Since the model-specific parameters can be chosen in any way, given $K=k$, we randomly choose $S$ to be one of the $\binom{3}{k}$ many models of size $k$ and that determines which entries from the full, $p=3$ dimensional coefficient vector $(0.5, 1, 5)^\top$ are active; 
we take the intercept and error variance fixed at 0 and 1, respectively.  With the columns of $x$ and all the model parameters determined, we can finally generate the response variable $Y$ according to \eqref{eq:lm}.  For each data set generated in this way, we evaluate $\pi_Y(S)$ to check if the true $S$ is contained in the confidence set $C_\alpha(Y)$ over a range of $\alpha$ values.  

For comparison, we also carry out a Bayesian analysis that uses the same prior $\prior_K$ above for the model complexity, a uniform prior on models of the given complexity, and then a conjugate normal--inverse gamma prior on the parameters, with standard ``non-informative'' choices for the hyperparameters.   Specifically, we take an inverse gamma prior, ${\sf InvGam}(0.01, 0.01)$, for the error variance $\Lambda$ and an independent normal conditional prior, ${\sf N}(0, 100\lambda)$, for the regression coefficients $\Phi$, given $\Lambda=\lambda$.  Then we compute the posterior distribution for the model $S$, obtain the corresponding $100(1-\alpha)$\% credible set, and check if it contains the true $S$, again over a range of $\alpha$ values.  

\begin{figure}[t]
    \centering
\scalebox{0.7}{\includegraphics{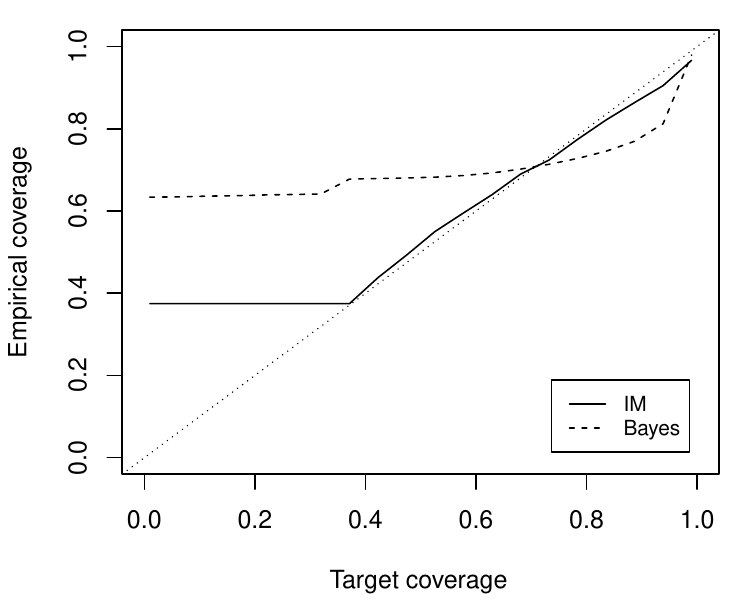}}
    \caption{Empirical versus target/nominal coverage for Bayes and IM uncertainty quantification for the experiment described in the main text.}
    \label{fig:coverage}
\end{figure}

Figure~\ref{fig:coverage} shows the realized empirical coverage probability for the two different confidence sets in the simulations compared to the target/nominal coverage probability.  On or above the diagonal lines implies that the empirical coverage probability attains or exceeds the target, which is what Corollary~\ref{cor:coverage} implies for the IM solution.  As discussed in Appendix~\ref{A:computation}, the computational strategy being used to evaluate the IM contour results in an apparently not-so-severe under-estimation---see the slight dip below the diagonal line in the upper right-hand corner of the plot in Figure~\ref{fig:coverage}.  The Bayes solution offers no coverage probability guarantees whatsoever and, as is apparent from the plot, its coverage probability can fall significantly short of the target level.

\section{Partial prior elicitation} 
\label{sec:prior_contour}

\subsection{General considerations}

There is an obvious question that needs to be addressed, that is, {\em where does the prior contour $q$ come from?}  To set the context of the discussion here, we admit that we cannot possibly know here what a data analyst will know in his/her application and, therefore, we are in no position to say what their prior possibility contour ought to be.  The goal here is simply to offer some different perspectives on how the data analyst might process the information they have available in order to determine a justifiable $q$.  We offer two different---but necessarily related---perspectives or approaches below.  What the two approaches have in common is that they aim to focus on what the data analyst would {\em do} rather than on what he/she {\em thinks}.  The two, of course, are related, but, as de Finetti argued, assessing agents' behavior is often easier than assessing their beliefs.  

For the sake of simplicity and concreteness, and without much loss of generality or practicality, we assume here that the data analyst has vacuous prior information about the specific model structure $S$ and about the structure-specific parameters $\Theta_S$.  On the other hand, we assume the data analyst does have genuine prior information about the size/complexity $K=|S|$ of the uncertain model structure $S$, e.g., the number of explanatory variables with nonzero effect size in linear regression.  Under this setup, i.e., where the only relevant information available concerns the complexity $K$, the prior possibility contour simplifies as 
\begin{align*}
q_{S,\Theta_S}(s,\theta_s) 
& = q_K(k) \, \underbrace{q_{S|K}(s \mid k)}_{\equiv 1} \, \underbrace{q_{\Theta_S|S}(\theta_s \mid s)}_{\equiv 1}.
\end{align*} 
That is, the prior possibility contour $q$ at a particular value $\theta$ of the uncertain $\Theta=(S,\Theta_S)$ depends only on the complexity $k = |S(\theta)|$ of the structure inherent in the particular $\theta$.  Consequently, the original question posed at the start of the section simplifies to the following: {\em how to choose the prior possibility contour $q_K$ for the uncertain complexity $K$?} 

\paragraph{Bayesian-style elicitation.} 
The conversion of relevant information about unknowns into precise prior probability distributions is a classical problem in Bayesian inference \citep[e.g.,][]{kass1995bayes, ohagan2004advanced, gelman2020holes, mikkola.etal.2024}.  The same considerations can and have been extended to cases where the goal is to determine imprecise probabilities; see \citet{smithson2014elicitation} and the references therein.  Recall that $q_K(k)$ is interpreted as an upper probability for the event ``$K=k$'' and, therefore, can be interpreted as an upper bound on the price one is willing to accept in exchange for an obligation to pay another player \$1 if it happens that $K$ indeed equals $k$.  These elicitation strategies are nontrivial to carry out in general, but our present situation is considerably simpler since we are only concerned with {\em a priori} uncertainty quantification about the bounded, nonnegative integer-valued unknown $K$.  So, it is not out of question that the data analyst can carry out such an assessment to obtain $q_K$.  

There are other ways to think about this assessment that do not strictly involve gambling.  Recall that the possibility contour $q_K$ has an interpretation similar to that of a p-value, i.e., $1-q_K(k)$ represents how surprised \citep[e.g.,][]{shackle1961decision} we would be if $K$ were equal to $k$, and, therefore, if one can directly assess this degree of surprise for each $k$, then $q_K$ obtains.  Indeed, there are some values of $K$ that would not be surprising at all, e.g., one's {\em a priori} ``best-guess,'' and these would be assigned $q_K = 1$; then there are ``extreme'' values of $K$, likely corresponding to the most-complex structures, that would be entirely surprising, and these are assigned $q_K \approx 0$.  There may be some intermediate values of $K$ that are surprising but not entirely so, and these would be assigned likewise intermediate values of $q_K$.  An alternative perspective on this is identifying the significance level at which one would reject the assertion ``$K=k$'' based on {\em a priori} information alone.  If one is uncomfortable with p-value-based reasoning, one can equivalently carry out this process by thinking in terms of confidence sets.  That is, if one can construct confidence sets $\{C_\alpha: \alpha \in [0,1]\}$ for $K$ based on prior information alone, then assign $q_K(k)$ equal to the largest $\alpha$ such that $C_\alpha \ni k$.  Note, though, \citet{juslin.etal.2007} and the references therein describe how (and why) confidence intervals constructed by individuals based on intuition tend to be overconfident in the sense of being too narrow.  This tendency is not too concerning since, in the present context, we are constructing the prior $q_K$ and not the final inferences.  

We consider a direct assessment of surprise, as described above, to be the gold-standard elicitation strategy, but there are other potentially simpler strategies that could be employed.  The angle we focus on here is one in which (subjective) probabilities are determined by drawing a connection between the context at hand and another more familiar context---or canonical example---where (objective) probabilities are available.  This is akin to the {\em constructive probability} arguments advanced in, e.g., \citet{shafer1985languages}.  In probability theory, it is possible that each possible realization of a random variable has small (or zero) probability.  For this reason, the notion of ``surprise'' in a particular value $k$ for random variables translates to the probability that $K$ is at least as extreme as $k$; again, we see the connection between the possibility contour and p-values.  A natural way to formulate this is to choose a probability distribution with mass function $f$ and then set 
\[ q_K(k) \gets \sum_{\kappa: f(\kappa) \leq f(k)} f(\kappa), \quad k=0,1,\ldots \]
To put this into practice, one considers different mass functions $f$---perhaps of a parametric form indexed by one or more hyperparameters---and chooses that for which the possibility assignments above are at least consistent with one's {\em a priori} assessment of surprise.  In our present context, one might entertain binomial or (truncated) geometric mass functions, indexed by a ``success probability'' hyperparameter, and then vary that hyperparameter until the resulting $q_K$ agrees with one's assessments.

\paragraph{Frequentist-style elicitation.} 
Classical methods do not (directly) make use of any prior information about the relevant unknowns.  So, no prior elicitation is required to define and implement the method itself.  Where the frequentist reveals his/her prior beliefs is in the investigation of their method's sampling properties.  (Bayesians actually do something similar; prior distributions are chosen such that the posterior inferences satisfy certain desirable frequentist properties.) The motivation behind discounting complex structures is that the true structure is believed to be rather simple.  As such, the theoretical investigations associated with lasso and other complexity-penalized estimators proceed under the assumption that the true structure is sufficiently simple, e.g., ``sparse''.  Accordingly, the partial-prior IM approach suggests a strategy by which the frequentist's prior beliefs about the structure $K$ can be assessed.  Treat $1-q_K(k)$ as the data analyst's degree of doubt that their IM is reliable in the sense of \eqref{eq:valid} when the true complexity is $k$.  For those $k$ such that the data analyst is not willing to tolerate any doubt, set $q_K(k)=1$.  Since the data analyst does not expect $K$ to be large, he/she should be willing to tolerate some degree of doubt that their desirable properties hold for extreme $k$, and it is through this assessment that the frequentist's $q_K$ emerges.  Since there is no free lunch, the data analyst must be willing to tolerate some degree of the aforementioned doubt; otherwise, he/she is forced to a vacuous prior and inference is inefficient in the sense that their IM favors the most complex structures.  We think this perspective is attractive---it guarantees that no one, neither the producer nor the consumer of frequentist methods, overlooks the fact that the efficiency gains associated with excluding the most complex structures come at the cost of a limited-complexity assumption about the true model structure.

\subsection{A simple strategy}
\label{SS:simple.strategy}

To keep the details as simple as possible, and again without much loss of generality or practicality, we focus here on a prior possibility contour of the form
\begin{equation}
\label{eq:contour.gamma}
q_K(k) = e^{-\gamma \, k}, \quad k=0,1,2,\ldots,p, \quad \gamma > 0. 
\end{equation}
The effect of different choices of the hyperparameter $\gamma$ will be investigated below.  Note that Bayesian methods in high-dimensional problems often take (marginal) prior distributions for the model complexity to be of the general form above; i.e., exponentially decreasing in $k$, and $\gamma$ is usually taken to be proportional to $\log p$, where $p$ is the (very large) number of explanatory variables \citep[see, e.g.,][]{castillo2015bayesian, martin.mess.walker.eb, belitser.ghosal.ebuq}.  This seemingly makes for a very strong penalty against complex model structures, but the ubiquity of these priors in the Bayesian literature suggests this form might in fact be necessary for posterior structure selection consistency.  Our present focus is on low- to moderate-dimensional cases, so we treat $\gamma$ itself as the hyperparameter and in this section we offer a simple strategy to help guide the choice of $\gamma$.  A more holistic approach that considers all of the available information in a given problem would be recommended, and we discuss this in the next section.  Our goal in this subsection is just to bring the partial-prior-dependent IM solution described above to life with as little overhead as possible.  

One of the most basic summaries of a probability distribution is its mean or expected value.  The prior contour $q_K$ has an associated credal set $\credal_K=\cred(q_K)$ for $K$ and, in turn, an upper expected value.  This upper expectation can again be computed using a Choquet integral.  Since the contour (\ref{eq:contour.gamma}) is monotone in the natural ordering on the integers, it turns out that the Choquet integral exactly agrees with the expected value under $\prior_K$, where $\prior_K$ is that ``most diffuse'' probability distribution in $\credal_K=\cred(q_K)$, i.e., 
\[ \overline{\text{mean}}(\gamma) = \sum_{k=0}^p k \, \prior_K(K=k), \quad \gamma > 0. \]
This is easy to evaluate numerically and plots of this upper expectation curve are shown in Figure~\ref{fig:upper.mean} for $p=3$ and $p=8$.  As expected, larger $\gamma$ makes the contour more concentrated at 0, hence the upper expectation should be decreasing in $\gamma$.  In the aforementioned Bayesian high-dimensional literature, the priors are likewise concentrated near 0 to the extent that the prior mean is rather small.  For example, the class of prior distributions considered in \citet{castillo2015bayesian} all have expected model complexity smaller than $O(p^{-1})$, which is quite small when $p$ is large.  Based on this, a reasonable strategy would be to set the upper expected value equal to something like $O(p^{-1})$ and then solve for $\gamma$ at least approximately.  The plots in Figure~\ref{fig:upper.mean} have a horizontal line drawn at $2p^{-1}$, which is not as aggressive as the prior assessments in the aforementioned literature, since we are not considering the high-dimensional situations here.  In the $p=3$ case this suggests choosing $\gamma$ close to 1 and, in the $p=8$ case, choosing $\gamma$ somewhere between 1.5 and 2.  Note that $\gamma \approx 1$ agrees with the choice taken in Section~\ref{SS:demo}.  Moreover, since $\log 3 \approx 1.1$ and $\log 8 \approx 2.08$, this strategy is suggesting a contour that is comparable to the penality employed by BIC.  

\begin{figure}[t]
\begin{center}
\subfigure[$p=3$]{\scalebox{0.6}{\includegraphics{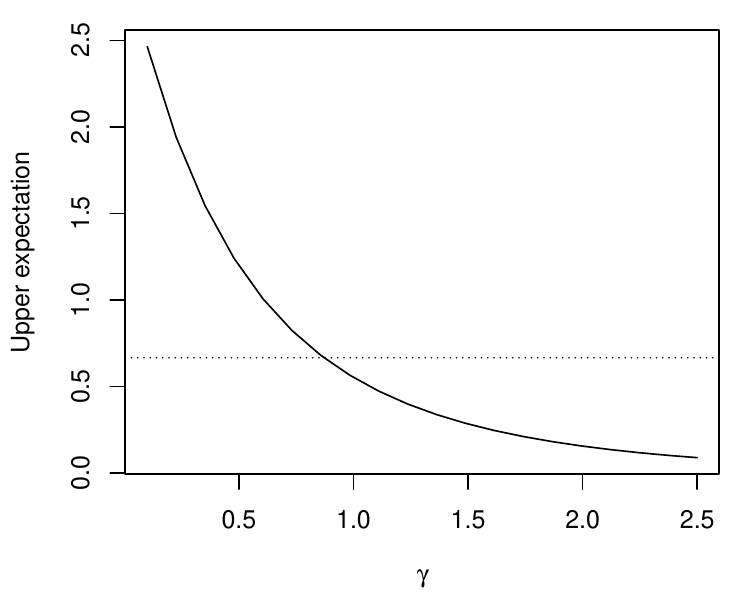}}}
\subfigure[$p=8$]{\scalebox{0.6}{\includegraphics{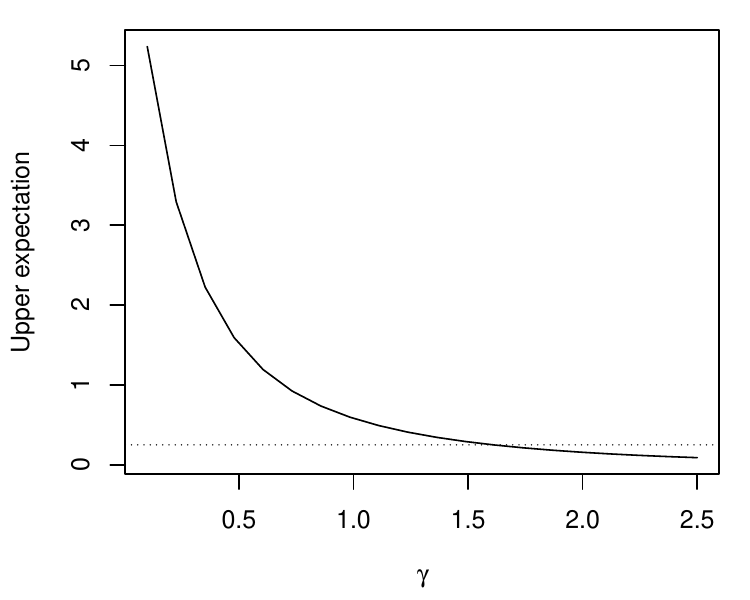}}}
\end{center}
\caption{Plots of the prior upper expected value $\gamma \mapsto \overline{\text{mean}}(\gamma)$ for two different values of $p$. }
\label{fig:upper.mean}
\end{figure}

A relevant refinement to the above analysis might treat the prior contour mode as a ``best prior guess'' based on available information about the problem, and then tune a scale parameter that controls how quickly the contour vanishes away from that mode.  

While the simple elicitation strategy just described is a reasonable one, we do not recommend its blind adoption as a practical default.  We cannot possibly know the specifics of each potential application, so practitioners should {\em always} take the full scope of the available information into consideration.  The discussions here and in the following section can help guide those considerations, nothing more.

\section{Applications} 
\label{sec:rda}

\subsection{Prostate cancer data}

In this section, we analyze the prostate cancer dataset from \citet{stamey1989prostate}, famously used as an illustration in \citet{tibshirani1996regression}. The study examined the relationship between the level of a prostate-specific antigen and eight clinical measures for 97 men about to receive a radical prostatectomy. The clinical measures included cancer volume (lcavol), prostate weight (lweight), age, amount of benign prostatic hyperplasia (lbph), an indicator for whether the cancer had spread to the seminal vesicles (svi), a measure of how much the cancer had penetrated the prostate wall (capsular penetration, lcp), a measure of cancer aggressiveness (gleason), and the percentage of tumor occupied by Gleason scores 4 and 5 (pgg45). Cancer volume, prostate weight, amount of benign prostatic hyperplasia, capsular penetration, and prostate-specific antigen were all log-transformed, hence the ``l'' in their acronym. The data along with descriptions of the clinical measures are available through the \verb|R| package \verb|faraway| \citep{faraway_rpackage}.

\begin{figure}[t]
\begin{center}
\subfigure[$\gamma=0.5$]{\scalebox{0.5}{\includegraphics{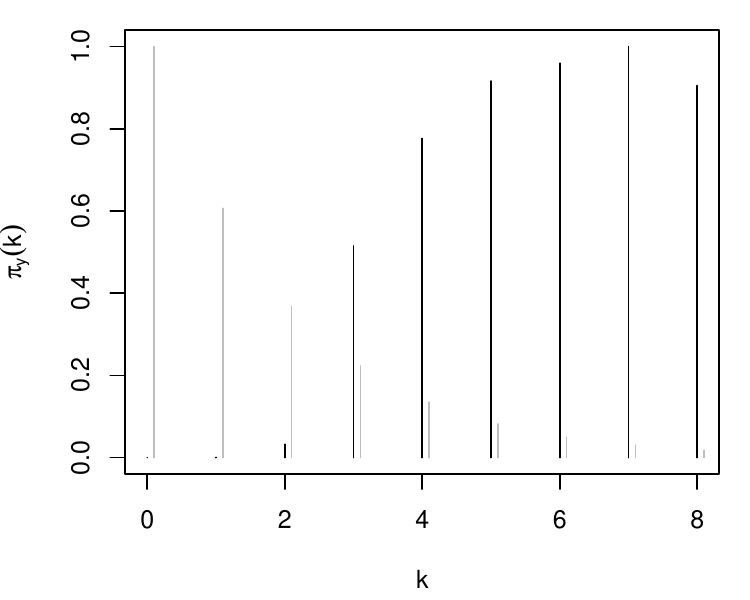}}}
\subfigure[$\gamma=1$]{\scalebox{0.5}{\includegraphics{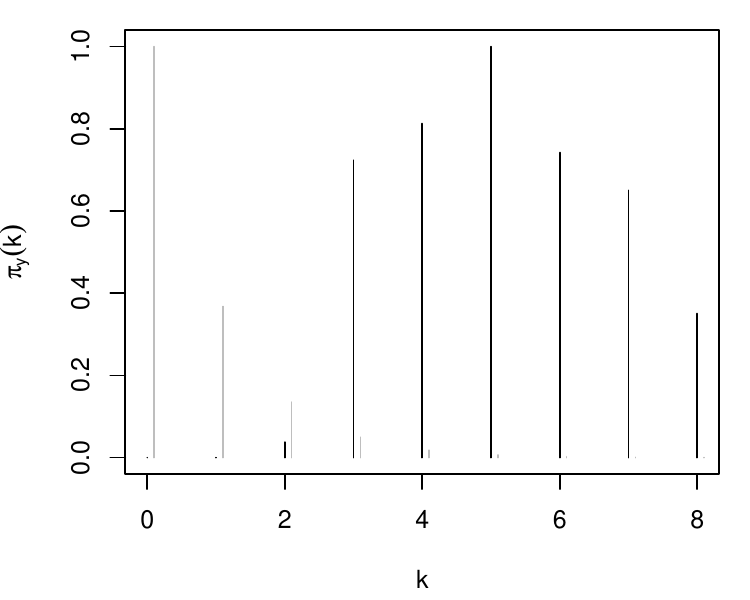}}}
\subfigure[$\gamma=1.5$]{\scalebox{0.5}{\includegraphics{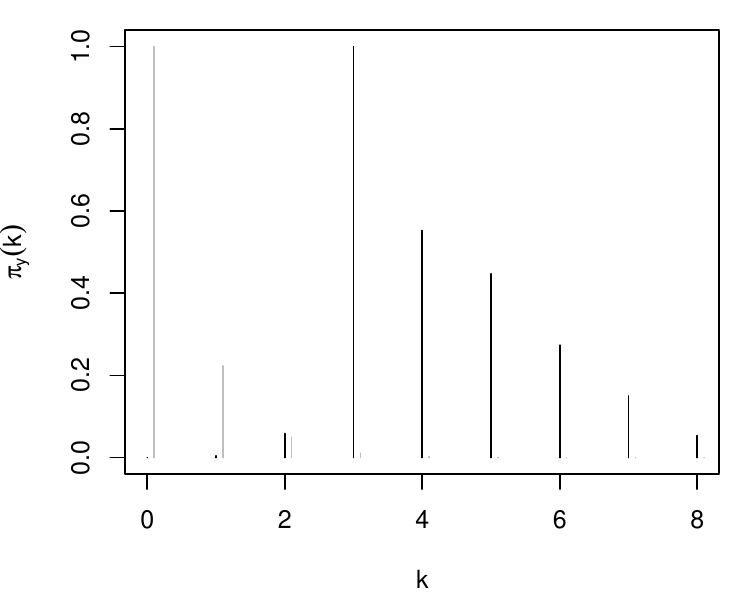}}}
\subfigure[$\gamma=2$]{\scalebox{0.5}{\includegraphics{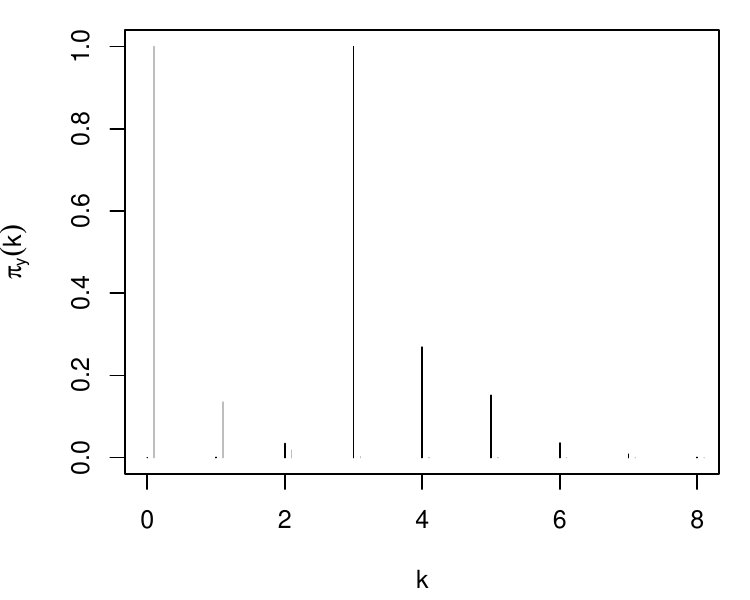}}}
\end{center}
\caption{Plots of the marginal IM contour $k \mapsto \pi_y^\text{marg}(k)$ for the model complexity $K=|S|$ in the prostate cancer data analysis based on four different values of the hyperparameter $\gamma$.  The slightly offset gray lines correspond to the prior contour $k \mapsto q_K(k)$.}
\label{fig:prostate}
\end{figure}

Figure~\ref{fig:prostate} shows the marginal IM contour $k \mapsto \pi_y^\text{marg}(k)$ for the model complexity $K=|S|$ based on these data and several different values of the hyperparameter $\gamma$.  Following the basic guidance in Section~\ref{SS:simple.strategy} above, we have reason to consider $\gamma$ values between 1.5 and 2.  For both of these endpoints, and all $\gamma$ values in-between, $\pi_y^\text{marg}$ has mode $k=3$, so this value would be contained in the IM's confidence sets for every confidence level.  Since there is only one structure $S$ with possibility $\pi_y(S)=1$ (see Table~\ref{table:prostate}), it is clear that this mode at $k=3$ uniquely corresponds to the structure that includes only the predictors lcavol, lweight, and svi, which agrees with the structure chosen by lasso \citep[][Table~1]{tibshirani1996regression}.  Of course, the larger the value of $\gamma$, the smaller the IM confidence set is, so the final choice comes down to a determination of how strongly one believes in sparsity of the underlying model.  When $\gamma=1.5$, the 95\% confidence set contains 43 model structures; when $\gamma=2$, the 95\% confidence set contains 9 model structures and the contents are displayed in Table~\ref{table:prostate}.  

\begin{table}[t]
\begin{center}
\begin{tabular}{ccccccccc}
\hline
lcavol & lweight & age & lbph & svi & lcp & gleason & pgg45 & $\pi_y(S)$ \\
\hline 
1 & 1 & 0 & 0 & 1 & 0 & 0 & 0 & 1.000 \\
1 & 1 & 1 & 0 & 1 & 0 & 0 & 0 & 0.156 \\
1 & 0 & 0 & 1 & 1 & 0 & 0 & 0 & 0.170 \\
1 & 1 & 0 & 1 & 1 & 0 & 0 & 0 & 0.265 \\
1 & 1 & 1 & 1 & 1 & 0 & 0 & 0 & 0.154 \\
1 & 1 & 0 & 0 & 1 & 1 & 0 & 0 & 0.101 \\
1 & 1 & 0 & 0 & 1 & 0 & 1 & 0 & 0.156 \\
1 & 1 & 0 & 0 & 1 & 0 & 0 & 1 & 0.188 \\
1 & 1 & 0 & 1 & 1 & 0 & 0 & 1 & 0.054 \\
\hline 
\end{tabular}
\end{center}
\caption{Contents of the IM's 95\% confidence set for the model structure in the prostate cancer data illustration.  The 0s and 1s in the first eight columns determine the contents of the model structure $S$.}
\label{table:prostate}
\end{table}

\subsection{World happiness data}

The World Happiness Report is conducted annually by the research-based consulting company Gallup.  The overarching goal of this project is to examine which variables are related to happiness, and to understand the strengths of those relationships.  The study consists of interviewing inhabitants in various countries throughout the world to assess their level of happiness, according to the Cantril Self-Anchoring Striving Scale \citep{glatzer2014}.  Following \citet{bergh.etal.2021}, we consider the data from 2018, available at \url{https://www.worldhappiness.report/ed/2018/}.  In addition to the aggregated happiness score within each country, the data set also reports on each country's log gross domestic product (lgdp), life expectancy (le), an assessment of social support (ss), an assessment of personal freedom (fr), an assessment of generosity (ge), and an assessment of government corruption (gc).  But it is expected that interactions between at least some of these variables would be relevant to describing variations in happiness, so we consider models that can contain any or all of those individual explanatory variables, along with pairwise interactions among logGDP, le, ss, and fr.  We adhere to the {\em principle of marginality} and only consider those model structures that include both main effects when the corresponding pairwise interaction is present.  This amounts to 440 different model structures to assess.

We applied the proposed IM solution to this problem almost exactly as described in Section~\ref{sec:add_ms} above.  That is, we formed a matrix $x$ of explanatory variables with a total of 12 columns: 6 for the main effects listed above and $\binom{4}{2}=6$ for the different pairwise interactions.  Then the relevant model structures consist of subsets of the 12 columns of $x$.  The one caveat, which creates a minor difference in the implementation here compared to what was described above, is that we restrict to those structures that respect the principle of marginality.  While this restriction is natural, it means that we cannot use the optimized {\tt leaps} package in \verb|R| which is tailored specifically to explore all subsets.  Instead, we carry out the partial-prior IM calculations using brute force enumeration of all 440 structures, which is no serious computational burden---the results described below were obtained in about 30 seconds on an ordinary Macbook Pro laptop computer.  

Just like in the prostate cancer data analysis above, Figure~\ref{fig:happy} shows the marginal IM contour $k \mapsto \pi_y^\text{marg}(k)$ for the model complexity $K=|S|$ based on these data, with prior hyperparameter $\gamma=2$, which is roughly the suggestion based on the general guideline in Section~\ref{SS:simple.strategy}.  With this choice, $\pi_y^\text{marg}$ clearly has mode $k=5$, so this value would be contained in the IM's confidence sets for every confidence level.  This corresponds to the model that includes the first four main effects---lgdp, le, ss, and fr---and the le $\times$ ss interaction.  This agrees with the maximum posterior probability model identified in the analysis presented in \citet{bergh.etal.2021}, but they can only offer a list of models ranked by their posterior probability; for example, they can offer no confidence set with frequentist-style coverage probability guarantees.  The proposed IM solution, on the other hand, offers provably reliable uncertainty quantification.  In particular, the corresponding 95\% confidence set for $S$ contains 10 model structures, and the specific contents are shown in Table~\ref{table:happy}.  

\begin{figure}[t]
\begin{center}
\scalebox{0.7}{\includegraphics{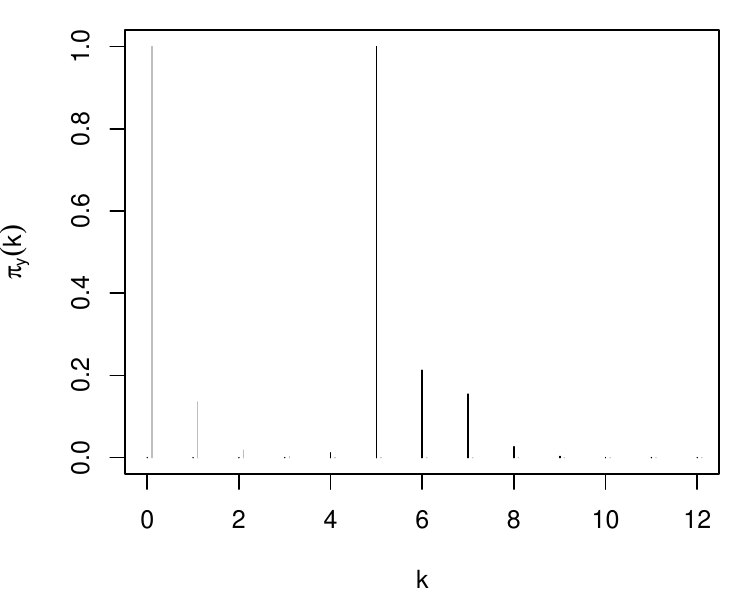}}
\end{center}
\caption{Plot of the marginal IM contour $k \mapsto \pi_y^\text{marg}(k)$ for the model complexity $K=|S|$ in the happiness data analysis with prior contour hyperparameter $\gamma=2$.  The slightly offset gray lines correspond to the prior contour $k \mapsto q_K(k)$.}
\label{fig:happy}
\end{figure}

\begin{table}[t]
\begin{center}
\begin{tabular}{ccccccccccccc}
\hline
lgdp & le & ss & fr & ge & gc & lgdp.le & lgdp.ss & lgdp.fr & le.ss & le.fr & ss.fr & $\pi_y(S)$ \\
\hline 
1 & 1 & 1 & 1 & 0 & 0 & 0 & 0 & 0 & 1 & 0 & 0 & 1.000 \\
1 & 1 & 1 & 1 & 0 & 0 & 0 & 0 & 0 & 1 & 0 & 1 & 0.213 \\
1 & 1 & 1 & 1 & 0 & 0 & 0 & 0 & 0 & 1 & 1 & 0 & 0.186 \\
1 & 1 & 1 & 1 & 0 & 0 & 0 & 1 & 0 & 1 & 0 & 0 & 0.181 \\
1 & 1 & 1 & 1 & 0 & 0 & 0 & 1 & 1 & 1 & 0 & 0 & 0.155 \\
1 & 1 & 1 & 1 & 0 & 0 & 0 & 1 & 0 & 1 & 0 & 1 & 0.130 \\
1 & 1 & 1 & 1 & 0 & 0 & 1 & 0 & 0 & 1 & 0 & 0 & 0.130 \\
1 & 1 & 1 & 1 & 1 & 0 & 0 & 0 & 0 & 1 & 0 & 0 & 0.123 \\
1 & 1 & 1 & 1 & 0 & 1 & 0 & 0 & 0 & 1 & 0 & 0 & 0.118 \\
1 & 1 & 1 & 1 & 0 & 0 & 0 & 1 & 0 & 1 & 1 & 0 & 0.110 \\
\hline 
\end{tabular}
\end{center}
\caption{Contents of the IM's 95\% confidence set for the model structure in the happiness data illustration; columns 6--12 correspond to the pairwise interactions.  The 0s and 1s in the first 12 columns determine the contests of the model structure $S$.}
\label{table:happy}
\end{table}

\section{Conclusion} 
\label{sec:conclusion}

This paper presents a new framework for provably reliable uncertainty quantification about a model's structure. In particular, we summarized the existing partial-prior IM framework which allows practitioners to accommodate any level of prior knowledge they have about their problem of interest, and includes the frequentist and Bayesian cases of no prior knowledge and exact prior knowledge as special cases. Previous work focused on uncertainty quantification about the model parameters for a given model structure. We introduced a new parameter for the model structure, constructed the marginal IM for it, and discussed its theoretical reliability properties---in particular, the marginal IM confidence sets for the model structure achieve a certain finite-sample coverage probability guarantee. 

There are a number of limitations associated with the proposed method, which naturally suggest open questions to be answered and directions to pursue as part of our future research. One main limitation is that computation of the IM contour can be expensive for even moderate $p$ as $2^p$ models need to be evaluated. In this paper, we used the {\tt leaps} package to efficiently calculate the model-specific residual sums of squares needed for estimating the IM contour at each of the $2^p$ models, and then applied the extension principle to formulate the marginal contour of the model size. It has recently come to our attention that \cite{bertsimas2016best} reformulated the best subset selection problem for linear regression as a mixed integer optimization problem, and demonstrated that they can handle $n$ in the 1000s and $p$ in the 100s in minutes with provable optimality; and near optimal solutions when $p > n$. This reformulation could be used to construct a profile-based IM for model structure that would be computationally feasible, and may be more efficient than the marginal contour for model structure constructed here. 
In addition to the computational challenges, the use of possibility contours for quantifying the prior information available is unfamiliar territory---not because translating prior information into a mathematical form is unusual, but because of the mathematical properties of the possibility contour. In particular, the requirement that the contour attain the value of one for at least one point in the space makes comparison between points in the space difficult. For example, constructing a prior contour that had small enough tails to appropriately penalize larger models (and in the extreme kill all covariates except for the intercept) required significant trial-and-error on simulated data not discussed here. The practitioner cannot be expected to perform this trial-and-error as it would invalidate the theoretical guarantees of the IM. As such, more effort is needed to simplify the elicitation process and help make users more comfortable in setting the prior contour from the outset.

\section*{Acknowledgments}

This work is partially supported by the U.S.~National Science Foundation, DMS--2412628.

\appendix

\section{Proofs from the paper}
\label{A:properties}

\begin{proof}[Proof of Theorem~\ref{thm:valid}]
This strong validity result is closely related to Theorem~1 in \citet{martin2022valid2}. The difference is that the present case involves marginalization over the nuisance,  structure-specific parameters, but this changes things only superficially.  Recall that $\overline\joint$ is defined as the upper envelope associated with the collection $\{\joint_\prior: \prior \in \credal\}$ so that
\[ \overline\joint\{ \pi_Y(S(\Theta) \leq \alpha \} = \sup_{\prior \in \credal} \joint_\prior\{ \pi_Y(S(\Theta)) \leq \alpha \}. \]
Next, recall that the contour $\pi$ is defined as 
\[ \pi_y(s) = \overline\joint\{ \eta_q(Y, S(\Theta)) \leq \eta_q(y,s)\} = \sup_{\prior' \in \credal} \joint_{\prior'} \{ \eta_q(Y, S(\Theta)) \leq \eta_q(y,s) \}, \]
where the use of $\prior'$ (instead of $\prior$) on the right-hand side is just to prepare for what comes next.  The key point is that, for any $\prior \in \credal$, 
\[ \pi_y(s) = \sup_{\prior' \in \credal} \joint_{\prior'} \{ \eta_q(Y, S(\Theta)) \leq \eta_q(y,s) \} \geq \joint_{\prior} \{ \eta_q(Y, S(\Theta)) \leq \eta_q(y,s) \}. \]
So, $\pi_y(s) \leq \alpha$ implies that 
\[ u_\prior(y,s) := \joint_{\prior} \{ \eta_q(Y, S(\Theta)) \leq \eta_q(y,s) \} \leq \alpha \quad \text{for all $\prior \in \credal$}. \]
Moreover, $u_\prior(Y, S(\Theta))$ is stochastically no smaller than $\unif(0,1)$ as a function of $(Y,\Theta) \sim \joint_\prior$ for each $\prior \in \credal$---this is by the general fact that plugging a random variable into its own distribution function returns a random variable stochastically no smaller than $\unif(0,1)$.  Therefore, 
\[ \overline\joint\{ \pi_Y(S(\Theta) \leq \alpha \} = \sup_{\prior \in \credal} \joint_\prior\{ \pi_Y(S(\Theta)) \leq \alpha \} \leq \sup_{\prior \in \credal} \joint_\prior\{ u_\prior(Y, S(\Theta)) \leq \alpha\} \leq \alpha, \]
which proves Theorem~\ref{thm:valid}'s claim.
\end{proof}

\begin{proof}[Proof of Corollary~\ref{cor:uvalid}]
This argument makes use of general lemma following Definition~3 in \citet{martin2022valid1}.  The claim is that the two events below are equivalent:
\begin{align*}
E_1 & = \{ (y,\theta): \text{$\uPi_y(H) \leq \alpha$ for some $H$ with $S(\theta) \in H$} \} \\
E_2 & = \{ (y,\theta): \pi_y(S(\theta)) \leq \alpha\}.
\end{align*}
On the one hand, it is easy to see that $E_1 \supseteq E_2$---just take $H = \{S(\theta)\}$.  On the other hand, if $(y,\theta) \in E_1$, so that $\uPi_y(H) \leq \alpha$ for some $H$ with $S(\theta) \in H$, then it follows by monotonicity of $\uPi_y$ that $\pi_y(S(\theta)) = \uPi_y(\{S(\theta)\}) \leq \uPi_y(H) \leq \alpha$, which implies $E_1 \subseteq E_2$.  Since $E_1$ and $E_2$ are equivalent, their $\overline\joint$-probabilities must be equal, which proves the claim of Corollary~\ref{cor:uvalid}. 
\end{proof}

\section{Computation-related details}
\label{A:computation}

Here we address the key question of how to evaluate the IM contour function defined in \eqref{eq:pi.s}.  A simple recasting of the general formula \eqref{eq:upper.joint} for probabilities with respect to the upper joint distribution for $(Y,\Theta)$ gives 
\begin{align*}
\pi_y(s) & = \overline\joint\bigl\{ \eta_q(Y,S(\Theta)) \leq \eta_q(y,s) \bigr\} \\
& = \sup_{\prior \in \credal} \int \underbrace{\prob_\theta\bigl\{ \eta_q(Y,S(\theta)) \leq \eta_q(y,s) \bigr\}}_{= g_{y,s}(\theta)} \, \prior(d\theta).
\end{align*}
The latter expression clearly corresponds to an {\em upper expectation} of $g_{y,s}(\Theta)$ with respect to the credal set of prior distributions $\prior$ for $\Theta$.  Since that function $g_{y,s}$ is bounded, and the credal set is induced by a possibility measure, there are alternative ways to express this upper expectation that are computationally simpler and tractable.  Indeed, Proposition~7.14 in \citet{lower.previsions.book} implies that $\pi_y$ can be equivalently expressed as 
\begin{equation}
\label{eq:im.contour}
\pi_y(s) = \int_0^1 \Bigl[ \sup_{\theta: q\{S(\theta)\} > w} \prob_\theta\bigl\{ \eta_q(Y,S(\theta)) \leq \eta_q(y,s) \bigr\} \Bigr] \, dw, \quad s \in \SS. 
\end{equation}
So, computation of the contour $\pi_y(s)$ involves an integral of a maximized probability.  The probability itself is easy to evaluate via Monte Carlo; the challenge is that the optimization requires scanning over the model structure space.  Fortunately, the {\tt leaps} package in \verb|R| \citep{leaps} makes this relatively easy in the linear regression setting under consideration here.  Since the prior contour $q$ only depends on the model parameter $\theta$ via the size/complexity of its inherent structure $s=S(\theta)$, and there are only finitely many such structures and sizes, the integrand in \eqref{eq:im.contour}, i.e., 
\[ w \mapsto \sup_{\theta: q\{S(\theta)\} > w} \prob_\theta \bigl\{ \eta_q(Y,S(\theta)) \leq \eta_q(y, s) \bigr\}, \]
is piecewise constant.  This makes the outer integration straightforward.  So, the remaining challenge is evaluating the finitely-many suprema in the above display.  

Of course, the above optimization can be done via a grid-search procedure, but this adds significant computational burden.  A key observation is that, for a given structure $S(\theta)=s$, the distribution of $\eta_q(Y,s)$ will be stochastically smaller---hence making the above probability larger---when the coefficients used to generate $Y$ are close or even equal to 0.  The intuition is that $\eta_q(Y,s)$ is roughly a ratio of scaled residual sums of squares.  First, the numerator involves a residual sum of squares corresponding to data being generated according to the model being fit, so this will not depend on the true regression coefficients.  Second, the denominator is a maximum over all models, and when we generate data with small/zero values for the regression coefficients, we expect roughly the same quality of fit regardless of whether the model is simple or complex.  In that case, the penalty term will end up deciding where the maximum in the denominator of $\eta_q(Y,s)$ is attained and, in particular, that maximum will be larger than the numerator, thus tending to make the ratio $\eta_q(Y,s)$ relatively small.  Moreover, since we are dealing with a ratio and since the error standard deviation is a scale parameter, we do not expect this to have strong/any effect.

Figure~\ref{fig:etacdf} shows the distribution function of $\eta_q(Y,s)$ for two different choices of $s$, in a simple case with only $p=3$ covariates, and for different choices on the magnitudes of the regression coefficients used in generating $Y$.  (In both cases, the error variance used in generating $Y$ is unity, and there was some not-severe variation across different values of the error variance, but the conclusion---that coefficients equal to 0 make $\eta_q(Y,s)$ smaller---remained unchanged.) The key observation is that the distribution function of $\eta_q(Y,s)$ is pointwise largest when the values of the regression coefficients are all 0 (black line), which confirms the above intuition. 

\begin{figure}[t]
\begin{center}
\subfigure[$s=(1, 0, 1)$]{\scalebox{0.6}{\includegraphics{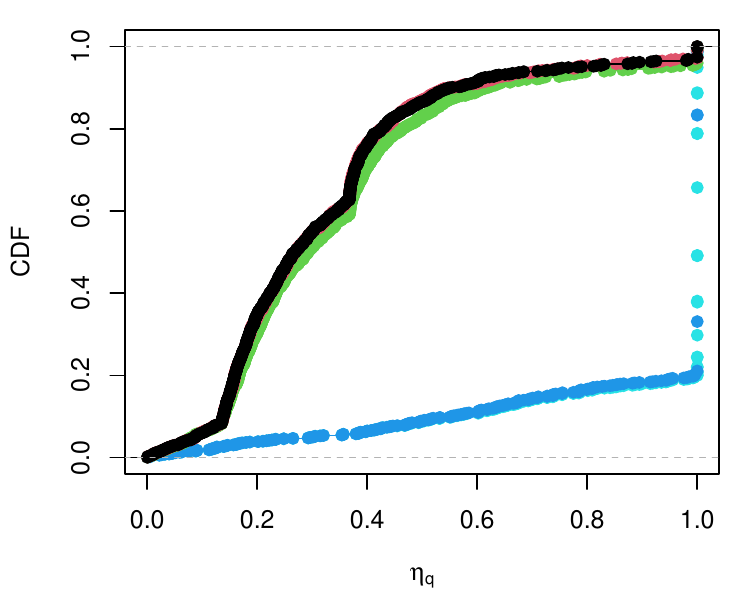}}}
\subfigure[$s=(0, 0, 1)$]{\scalebox{0.6}{\includegraphics{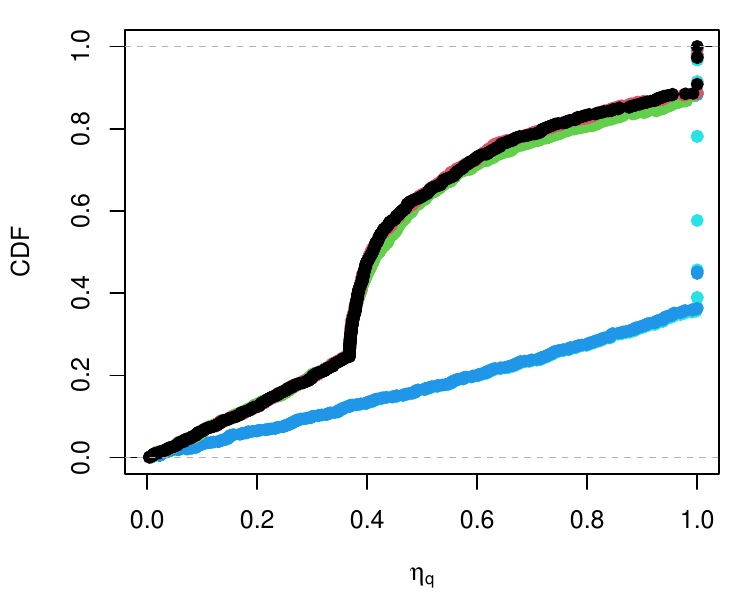}}}
\end{center}
\caption{Empirical distribution functions of $\eta_q(Y,s)$ for two different $s$ and several different sets of true regression coefficients (distinguised by color).  Black line is the distribution function corresponding to regression coefficients all equal to 0.}
\label{fig:etacdf}
\end{figure}

In light of the above observation, we carry out the IM computations with the integrand in the previous display replaced by 
\begin{equation}
\label{eq:approx}
w \mapsto \max_{r \in \SS: q(r) > w} \prob\{ \eta_q(Z, r) \leq \eta_q(y,s)\},
\end{equation}
where the probability ``$\prob$'' assumes that $Z$ is filled with iid standard normal random variables.  This is relatively easy to evaluate via Monte Carlo; indeed, the distribution of $\eta_q(Z,r)$, for each model structure $r$, can be precomputed independent of the data $y$ and of the posited structure $s$, after which the evaluation of the contour function $s \mapsto \pi_y(s)$ is very easy.  It is important to emphasize that this is indeed an approximation and, in particular, that it results in an under-estimate of the contour $\pi_y(s)$ for which the theoretical properties in Section~\ref{SS:valid} are established.  Our experience suggests that the approximation error is insignificant, but there is nonetheless some error.  

A more rigorous justification for the conservativeness of a similar approximation to (\ref{eq:approx}), using the profile-based marginal IM possibility contour, is as follows.
\begin{align*}
\eta_q^\text{\sc pr}(y,s) & = \frac{L_y(s,\hat\theta_s) \, q(s)}{\max_{r \in \SS} L_y(r,\hat\theta_r) \, q(r)}, \quad s \in \SS \\
& = \frac{\|y - (x_s^\top x_s)^{-1} x_s^\top y\|^{-n} \, q(s)}{\max_{r \in \SS} \|y - (x_r^\top x_r)^{-1} x_r^\top y\|^{-n} \, q(r)} \\
& = \frac{\|(I_n - H_s)y\|^{-n} \, q(s)}{\max_{r \in \SS} \|(I_n - H_r)y\|^{-n} \, q(r)},
\end{align*}
where $H_s := x_s (x_s^\top x_s)^{-1} x_s^\top$.  Consider the profile-based marginal IM possibility contour,
\[
\pi_y^\text{\sc pr}(s) = \sup_{\theta_s}\prob_{s,\theta_{s}}\{ \eta_q^\text{\sc pr}(Y,s) \leq \eta_q^\text{\sc pr}(y,s) \}.
\]
The $\sup_{\theta_s}$ calculation can be avoided with the following approximation.  In the probability statement above, $Y = x_s\phi_s + \lambda^{1/2} Z$ with $Z \sim {\sf N}_n(0,I_{n})$, and so for any $r \in \SS$,
\[
\|(I_n - H_r)Y\|^2 - \|(I_n - H_r)Z\|^2 \lambda = \|(I_n - H_r)x_s\phi_s\|^2 + 2\lambda^{1/2} \phi'_s x'_s(I_n - H_r) Z.
\]
Note that in the case $r = s$, the right side is zero, meaning that $Y$ in the numerator of $\eta_q^\text{\sc pr}(Y,s)$ can be replaced by $\lambda^{1/2}Z$, i.e., not depending on $\phi_s$.  

Next, define the event
\begin{align*}
A &:= \bigcap_{r\in\SS}\Big\{ \|\lambda^{1/2}(I_n - H_r)Z\|^{-n} \ge \|(I_n - H_r)Y\|^{-n} \Big\} \\
& \subseteq \bigg\{ \frac{1}{\max_{r \in \SS} \|\lambda^{1/2}(I_n - H_r)Z\|^{-n} \, q(r)} \le \frac{1}{\max_{r \in \SS} \|(I_n - H_r)Y\|^{-n} \, q(r)} \bigg\}.
\end{align*}
Then
\begin{align*}
\pi_y^\text{\sc pr}(s) & = \sup_{\theta_s}\prob_{s,\theta_{s}}\{ \eta_q^\text{\sc pr}(Y,s) \leq \eta_q^\text{\sc pr}(y,s) \} \\
& = \sup_{\theta_s}\bigg[\prob_{s,\theta_{s}}\{ \eta_q^\text{\sc pr}(Y,s) \leq \eta_q^\text{\sc pr}(y,s), A \} + \prob_{s,\theta_{s}}\{ \eta_q^\text{\sc pr}(Y,s) \leq \eta_q^\text{\sc pr}(y,s), A^c\} \bigg] \\
& \le \prob_{s}\{ \eta_q^\text{\sc pr}(Z,s) \leq \eta_q^\text{\sc pr}(y,s) \} + \sup_{\theta_s}\prob_{s,\theta_{s}}(A^c),
\end{align*}
where
\begin{align*}
\prob_{s,\theta_{s}}(A^c) & \le \sum_{r\in\SS}\prob_{s,\theta_{s}}\Big\{ \|\lambda^{1/2}(I_n - H_r)Z\|^{-n} < \|(I_n - H_r)Y\|^{-n} \Big\} \\
& = \sum_{r\in\SS\setminus\{s\}}\prob_{s,\theta_{s}}\Big\{ \|\lambda^{1/2}(I_n - H_r)Z\|^2 > \|(I_n - H_r)Y\|^2 \Big\} \\
& = \sum_{r\in\SS\setminus\{s\}}\prob\Big\{ -2\lambda^{1/2} \phi'_s x'_s(I_n - H_r) Z > \|(I_n - H_r)x_s\phi_s\|^2 \Big\} \\
& = \sum_{r\in\SS\setminus\{s\}} F_{\text{N}(0,1)}\Big\{ -\|(I_n - H_r)x_s\phi_s\| / (2\lambda^{1/2}) \Big\}.
\end{align*}
As long as $x$ is full rank, it can be expected that $\|(I_n - H_r)x_s\phi_s\| = O(n)$, implying that $\prob_{s,\theta_{s}}(A^c)$ is close to zero for large signal-to-noise ratio and/or for large $n$.  Thus, $\pi_y^\text{\sc pr}(s)$ can be conservatively approximated by $\prob_{s}\{ \eta_q^\text{\sc pr}(Z,s) \leq \eta_q^\text{\sc pr}(y,s) \}$, ensuring validity.

\bibliographystyle{apalike}
\bibliography{im_ref}

\end{document}